\documentclass{article}
\usepackage{titlefoot}
\usepackage{authblk}
\usepackage[utf8]{inputenc}
\usepackage[english]{babel}
\usepackage{microtype} %
\usepackage{xspace, url, array, float, xcolor}
\usepackage[scriptsize,raggedright,hang,tight]{subfigure} %
\usepackage{amsmath} %
\usepackage{amsthm} %
\usepackage{amssymb} %
\usepackage{wasysym} %
\usepackage{graphicx}
\usepackage{natbib}

\newcommand{\set}[1]{\ensuremath{\{#1\}}}
\newcommand{\setwith}[2]{\ensuremath{\set{{#1}\,\mid\,{#2}}}}
\newcommand{\group}[1]{\ensuremath{\!\left(#1\right)}}
\newcommand{\fundef}[3]{\ensuremath{{#1}\colon{#2}\to{#3}}}
\newcommand{\arc}[2]{\ensuremath{(#1, #2)}}
\newcommand{\edge}[2]{\ensuremath{\set{#1, #2}}}
\newcommand{\inarc}[1]{\ensuremath{\delta^-({#1})}}
\newcommand{\inarcX}[1]{\ensuremath{\tilde\delta^-({#1})}}
\newcommand{\Loss}{\ensuremath{\textit{Loss}}}
\newcommand{\win}{\ensuremath{\textit{win}}}
\newcommand{\save}{\ensuremath{\textit{save}}}
\newcommand{\SteinLib}{\textsc{SteinLib}\xspace}

\newcommand{\bigO}[1]{\ensuremath{\mathcal{O}\group{#1}}}

\newcommand{\Real}{\ensuremath{\mathbb{R}}}
\newcommand{\RealNonNegative}{\ensuremath{\Real_{\geq 0}}}

\newcommand{\eg}{e.g.,  }
\newcommand{\ie}{i.e.,  }
\newcommand{\wrt}{w.r.t.\xspace}

\newcommand{\complexityclass}[1]{\textsf{#1}}
\newcommand{\cP}{\complexityclass{P}\xspace}
\newcommand{\cNP}{\complexityclass{NP}\xspace}

\newcommand{\cMAXSNP}{\complexityclass{MAX\,SNP}\xspace}

\newcommand{\op}[1]{\textrm{#1}}
\newcommand{\MST}[1]{\ensuremath{\op{MST}(#1)}}
\newcommand{\Vor}[1]{\ensuremath{\mathcal{V}(#1)}}

\newcommand{\algo}[1]{\textsf{#1}}

\newcommand{\linprog}[3]{%
 \begin{subequations}
  \begin{alignat}{2}%
   \text{#1} \quad & #2 \\%
   \text{s.\,t.} \quad #3%
  \end{alignat}
 \end{subequations}
}
\newcommand{\constraint}[4]{ & #1 \, #2 \, #3 , & \quad & \text{#4} }

\ifx\todo\undefined\newcommand{\todo}[1]{\xspace{\bfseries\sffamily\textcolor{red}{[{#1}]}}\xspace}\fi

\newtheorem{theorem}{Theorem}
\newtheorem{definition}[theorem]{Definition}
\newtheorem{observation}[theorem]{Observation}
\newtheorem{lemma}[theorem]{Lemma}

\newcommand{\tbl}[2]{%
  \caption{#1}

  \centering
  #2
}

\bibpunct{[}{]}{,}{a}{}{;}

\title{Strong Steiner Tree Approximations in Practice}

\author{Stephan Beyer and Markus Chimani}
\affil{\small Institute of Computer Science, Osnabrück University\\
 \texttt{\{stephan.beyer,markus.chimani\}@uni-osnabrueck.de}}
\date{}

\begin{document}

\maketitle

\begin{abstract}
 In this experimental study
  we consider Steiner tree approximation algorithms
  that guarantee a constant approximation ratio smaller than~$2$.
 The considered greedy algorithms
  and approaches based on linear programming
  involve the incorporation
  of $k$-restricted full components for some $k \geq 3$.
 For most of the algorithms,
  their strongest theoretical approximation bounds
  are only achieved for $k \to \infty$.
 However,
  the running time is also exponentially dependent on $k$,
  so only small $k$ are tractable in practice.

 We investigate different implementation aspects
  and parameter choices
  that finally allow us to construct algorithms
  (somewhat)
  feasible for practical use.
 We compare the algorithms
  against each other,
  to an exact LP-based algorithm,
  and to fast and simple $2$-approximations.
\end{abstract}

\unmarkedfntext{
 Funded by the German Research Foundation~(DFG), project number CH 897/1-1.\\
 Preliminary versions of this work appeared as
  \citep{CW11}
  and
  \citep{BC14}.
}

\section{Introduction}

The Steiner tree problem,
 essentially asking for the cheapest connection of points
  in a metric space,
 is a fundamental problem in
 computer science and operations research.
In the general setting,
 we are given a connected graph $G = (V, E)$
 with edge costs $\fundef{d}{E}{\RealNonNegative}$
 and a subset $R \subseteq V$ of nodes.
Those \emph{required} nodes $R$ are called \emph{terminals},
 and $V \setminus R$ are called \emph{nonterminals}.
A terminal-spanning subtree in $G$
 is called a \emph{Steiner tree}.
The \emph{minimum Steiner tree problem in graphs~(STP)}
 is to find a Steiner tree $T = (V_T, E_T)$ in $G$
 with $R \subseteq V_T \subseteq V$
 and
 minimizing the cost $d(T) := d(E_T) := \sum_{e \in E_T}{d(e)}$.

As one of the hard problems identified by \citet{K72},
 the STP is even \cNP-hard for special cases like
 bipartite graphs with uniform costs~\citep{HRW92}.
\citet{PY91} proved that the STP is in \cMAXSNP,
 \citet{BP89} showed that this is even the case
  if edge costs are limited to $1$ and $2$.
No problem in \cMAXSNP has a polynomial-time approximation scheme~(PTAS)
 if $\cP \neq \cNP$,
 that is,
 it cannot be approximated arbitrarily close to ratio $1$
 in polynomial time
 under widespread assumptions.
The best known lower bound for an approximation ratio is
 $96/95 \approx 1.0105$~\citep{CC08}.

The STP has various applications
 in the fields of
  VLSI design,
  routing,
  network design,
  computational biology,
  and
  computer-aided design.
It serves as a basis for generalized problems like
 prize-collecting
 and stochastic Steiner trees,
 Steiner forests,
 Survivable Network Design problems,
 discount-augmented problems like Buy-at-Bulk or Rent-or-Buy,
 and appears as a subproblem in problems like the Steiner packing problem.

The versatile applicability of the STP
 gave rise to a lot of research from virtually all algorithmic points of view:
  heuristics,
  metaheuristics,
  and
  approximation algorithms
  \citep{R83,KR92,GN12,%
   ARUW01,LLLR14,%
   TM80,KMB81,M88,Z92,Z93:IPL,Z93:AI,BR94,GW95,Z95,KZ95,PS97,HP99,RZ05,BGRS13,GORZ12%
  }
 on the one side
 and exact algorithms based on
  branch-and-bound,
  dynamic programming,
  fixed-parameter tractability,
  and integer linear programs
  \citep{%
   SFG82,%
   DW72,HSV14,V11,%
   CMZ12,FBN13,%
   A80,W84,PV01%
  }
  on the other side.
Both areas are complemented by research about reduction techniques on the problem \citep{DV89,PV02}.
However, by far not all of that research is backed by experimental studies.
This paper
 attempts to close this gap in the field of strong approximation algorithms
 where there has only been the preliminary conference papers leading to this article \citeyearpar{CW11,BC14}, and the work of \citet{CGSW14}
 during the \emph{11th DIMACS Implementation Challenge} \citeyearpar{DIMACS14Bounds}.
By \emph{strong} approximations we denote
  approximations with ratio smaller than $2$.
Although those algorithms
 have been a breakthrough in theory,
 their actual practicability has remained unclear.
We contribute
 by implementing, extending, evaluating, and comparing
 the different strong algorithms
 and a variety of algorithmic variants thereof.
We also compare them to simple $2$-approximations and exact algorithms.

In the following section
 we first give an overview on the basic ideas behind strong algorithms
 and their evolution.
We then describe the two known classes of algorithms,
 greedy combinatorial algorithms
 and LP-based algorithms,
 in more detail.
We only give as many details as
 necessary to understand
 our subsequent design choices
 and algorithm variants.
Section~\ref{sec:alg-eng}
 is about different practical variants
 for the algorithms.
In Section~\ref{sec:exp-eval}
 we evaluate the algorithms and their variants,
 and compare their running times and solution qualities
 to basic $2$-approximations and an exact algorithm.

\section{The Algorithms}
\label{sec:algorithms}

For any graph $H$,
 we denote
  its nodes by $V_H$,
  its edges by $E_H$,
  and its terminals by $R_H$.
When referring to the input graph $G$,
 we omit the subscript.
By $\MST{H}$ we denote a minimum spanning tree in $H$.
Let $\bar{G}$ be the metric closure of $G$,
 that is,
 the complete graph on $V$
 such that the cost of each edge $\edge{u}{v}$
  is the minimum cost of any $u$-$v$-path in $G$.
Let $\bar{G}_U$ denote the $U$-induced subgraph of $\bar{G}$ for $U \subseteq V$.
For any graph $G$ and any subgraph $H \subseteq G$,
 we denote by $G/H$
 the result of contracting $H$ into a single node in $G$.

We first give an overview on purely combinatorial approximation algorithms
 for the STP describing the basic ideas coarsely.
In Section~\ref{sec:GCF} we provide a more detailed description
 of some of the formerly strongest approximation algorithms.
Section~\ref{sec:LP} is about recent approximation algorithms
 that are based on linear programming techniques.

The simplest algorithms are basic $2$-approximations.
The algorithm by \citet{TM80}
 can be compared to the Jarn\'ik-Prim algorithm
 to find minimum spanning trees.
In each iteration, a shortest path to an unvisited terminal
 (instead of a single edge to an unvisited node) is added.
That is, the algorithm builds the Steiner tree
 starting with a single terminal node
 and iteratively adds the shortest path to the nearest terminal to the tree.
The algorithm by \citet{KMB81}
 computes $\bar{G}_R$ and $\MST{\bar{G}_R}$.
After replacing the edges of $\MST{\bar{G}_R}$ by the corresponding shortest paths in $G$,
 and cleaning up the obtained graph
 (\ie
  breaking possible cycles
  and pruning Steiner leaves),
 we obtain a Steiner tree that is a $2$-approximation.
\citet{M88} suggests a more time-efficient variant that
 exploits the use of Voronoi regions.

\begin{figure}
 \centering
 \subfigure[The instance.]{\scalebox{0.9}{%
  \includegraphics{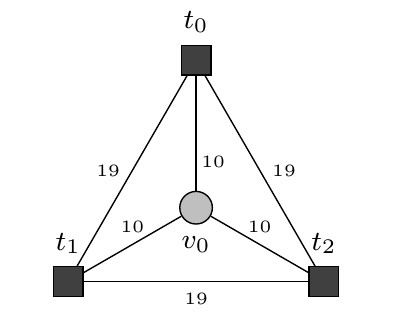}
 }}
 \subfigure[Possible result of $2$-approximations.]{\scalebox{0.9}{%
  \includegraphics{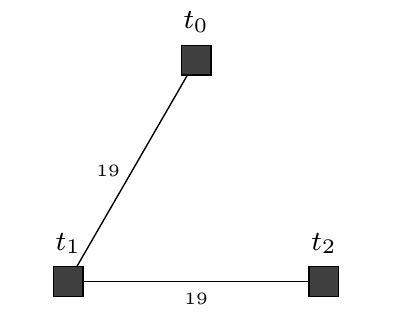}
 }}
 \subfigure[Improved approximation (and optimum solution).]{\scalebox{0.9}{%
  \includegraphics{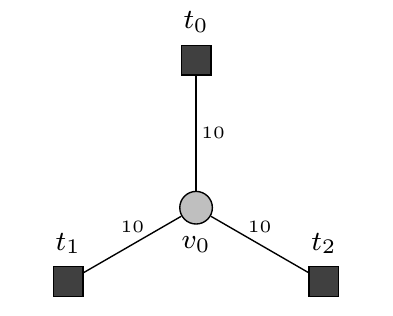}
 }}
 \caption{%
  A simple example that shows how converting a well-chosen nonterminal (here $v_0$)
  to a terminal can improve a $2$-approximation.
 Square nodes are terminals, circular nodes are nonterminals.
 }
 \label{fig:simple-improvement}
\end{figure}

Assume we want to improve a Steiner tree $T$
 that is obtained by a $2$-approximation.
One idea
 is to find nonterminals
 that are not already included in $T$
 but whose inclusion would improve $T$.
Hence, by temporarily converting these nonterminals to terminals
 and applying a $2$-approximation on the new instance,
 the result can be better.
Fig.~\ref{fig:simple-improvement} shows such a case.
The crucial ingredient in this idea is the choice
 of the nonterminals.
Zelikovsky~\citeyearpar{Z92,Z93:IPL,Z93:AI} gives an approach
 that guarantees an approximation ratio of $11/6$.
For every choice of three terminals,
 his algorithms find a nonterminal to be chosen
 as the center of a star
 where the three terminals are the leaves.
Among all these stars,
 the ``best'' ones%
 ---according to some greedy criterion function---%
 are chosen for the Steiner tree to be constructed.

From another point of view,
 this approach exploits the decomposition of a Steiner tree
 into \emph{full components},
 a concept already mentioned by \citet{GP68}.
A Steiner tree is \emph{full} if its set of leaves coincides with its set of terminals.
Any Steiner tree can be uniquely decomposed into full components
 by splitting up inner terminals.
We say a \emph{$k$-restricted component} is a full component with at most $k$ leaves
 and a \emph{$k$-component} is a full component with exactly $k$ leaves.
A \emph{$k$-restricted Steiner tree} is a Steiner tree
 where each component is $k$-restricted.
In these terms,
 the algorithm by \citet{KMB81}
 is a Steiner tree construction using $2$-components,
 and the mentioned algorithms by Zelikovsky use $3$-restricted components.

All strong algorithms for the STP
 known so far
 exploit the decomposition of a Steiner tree into full components
 by first constructing a set of $k$-restricted components
 and then putting the full components together to obtain a $k$-restricted Steiner tree.
Interestingly,
 the cost ratio between a minimum $k$-restricted Steiner tree and a minimum Steiner tree
 is (tightly) bounded by
 $\varrho_k := 1 + \frac{2^r}{(r-1) \, 2^r + k}$
 with $r = \lfloor \log_2{k} \rfloor$
  \cite{BD95}.\footnote{%
  Note that $k$-restricted components are always constructed based on $\bar{G}$.
  If they were based on $G$,
   a $k$-restricted component may not even exist, and, if it exists, the ratio is unbounded.
 }
Note that $\varrho_2 = 2$
 is also the approximation ratio for approximations based on $2$-restricted Steiner trees
 since the minimum $2$-restricted Steiner tree of a graph corresponds to
 the back-transformation of $\MST{\bar{G}_R}$.
However, $\varrho_k$ for $k \geq 3$
 cannot simply serve as an approximation ratio:
 it is not known whether a polynomial-time algorithm for $k=3$ exists.
However, there is a PTAS for this case \citep{PS97},
 so it is possible to approximate arbitrarily close to ratio $\varrho_3 = 5/3$.
Obtaining a minimum $k$-restricted Steiner tree
 for $k \geq 4$ is strongly \cNP-hard,
 as follows from a trivial reduction from \textsc{Exact Cover by $r$-Sets} with $r = k - 1$.

With respect to $\varrho_k$,
 Zelikovsky's approach yields an approximation ratio of $\frac{\varrho_2 + \varrho_3}{2} = \frac{11}{6}$.
\citet{BR94}
 were the first to generalize this approach to arbitrary $k$
 by using rather complicated preselection and construction phases.
They obtain a ratio of $\varrho_2 - \sum_{i=3}^{k}\frac{\varrho_{i-1} - \varrho_i}{i-1} \geq 1.7333$
 but, in particular, $11/6 \approx 1.8333$ for $k = 3$
 and $16/9 \approx 1.7778$ for $k = 4$.
\citet{Z95} generalizes his former approach using another greedy selection criterion
 (the \emph{relative greedy heuristic})
 and obtains an approximation ratio of $(1 + \ln{\frac{\varrho_2}{\varrho_k}}) \varrho_k \approx (1.693 - \ln\varrho_k) \varrho_k$
 which becomes approximately $1.693$ for $k \to \infty$ since $\varrho_k$ tends to~$1$.
However, the proven approximation ratios for $k = 3, 4$ are only about $1.97$ and $1.93$, respectively.

\citet{KZ95} introduce the notion of the \emph{loss}
 of a full component
 to allow some more sophisticated preprocessing.
They utilize it to prove small improvements for
 the Berman-Ramaiyer algorithm with $k = 4$ (from $1.778$ to $1.757$)
 and for the relative greedy heuristic with $k \to \infty$ (from $1.693$ to $1.644$).
\citet{HP99}
 use the idea of Karpinsky and Zelikovsky
 in an iterated manner.
They incorporate the loss of a full component with some well-chosen weight
 into the relative greedy heuristic
 and solve it to obtain a Steiner tree.
In each iteration, the weight is decreased
 and the modified relative greedy heuristic is run again.
The optimal sequence of weights can be found using numerical optimization.
For 11~iterations and $k \to \infty$,
 they obtain an approximation ratio of $1.598$.

\citet{PS97}
 use algebraic techniques to attack the problem of obtaining a minimum $3$-restricted Steiner tree.
They obtain a randomized fully polynomial-time $(\varrho_3 + \varepsilon)$-approximation scheme,
 however,
 with a sequential time complexity of $\bigO{\frac{\log(1/\varepsilon)}{\varepsilon} n^{11 + \omega} \log{n}}$
 where $\omega$ is the exponent of matrix multiplication.

The so far best purely combinatorial approximation algorithm is the \emph{loss-contracting algorithm} by
 \citet{RZ05}.
The obtained approximation ratio is $(1 + \frac{1}{2}\ln(\frac{4}{\varrho_k} - 1)) \varrho_k$
 which tends to $1.549$ for $k \to \infty$
 and is $1.947$ and $1.883$ for $k = 3,4$, respectively.
We will describe it in more detail in the following section,
 before discussing the even stronger LP-based algorithms in Section~\ref{sec:LP}.

\subsection{Greedy Contraction Framework}
\label{sec:GCF}

A lot of the strong algorithms
 are based on the contraction of full components.
The idea behind all these algorithms is basically the same.
It was first summarized by \citet{Z95}
 and called the \emph{greedy contraction framework}~(GCF).

We are (implicitly or explicitly) given a list $\mathcal{C}_k$ of $k$-restricted components
 and a \emph{win function} $\win_f$
 that characterizes the benefit of choosing a full component in the final Steiner tree.
Its value for a specific full component $C$ is called \emph{win},
 and we call $C$ \emph{promising}
 if choosing $C$ guarantees an improvement.
The GCF begins by computing the metric closure $M := \bar{G}_R$ over
 the terminals $R$ in $G$.
Recall that
 deducing a Steiner tree from $\MST{\bar{G}_R}$
 yields a $2$-approximation.
Iteratively,
 the GCF finds a full component $C \in \mathcal{C}_k$
  that maximizes $\win_f(M, C)$,
 and contracts $C$ in $M$
 if this win is promising.
This process is repeated as long as there are full components with promising wins.

Each time a full component is contracted,
 it is incorporated into the final Steiner tree.
This can be done
 by
  converting the nonterminals of the chosen full components into terminal nodes
  and computing a $2$-approximation using the new terminal node set.
 Alternatively,
  we can start with an empty graph $T$,
  inserting each chosen full component into $T$,
  and finally returning $\MST{T}$
  to clean up cycles that may have arisen
  (see \citep{Z95,RZ05}).

The \emph{loss-contracting algorithm}~(LCA) by \citet{RZ05}
 is a variant of GCF
 with a small difference:
 not the whole full component is contracted
 but only its loss.
\begin{definition}[core edges, loss]
 \label{def:loss-core}
 We denote as \emph{core edges}
  of $C$ a minimal subset of $E_C$
  whose removal disconnects all terminals in $C$.
 The \emph{loss}
  $\Loss(C)$
  of a full component $C$
  is the minimum-cost subforest of $C$
  such that all inner nodes are connected to leaves.
\end{definition}
A full component with $k$ leaves has exactly $k-1$ core edges.
Note that the definition of core edges
 does not involve edge costs.
The complement of any spanning tree in $C/R_C$
 forms a set of core edges.
$\Loss(C)$, however, is a \emph{minimum} spanning tree in $C/R_C$.
Overall, the set of non-loss edges (the complement of $\Loss(C)$)
 is one possible core edge set
 (but not the only one).
This distinction will become relevant
 for the randomized algorithm described in Section~\ref{sec:LP}.
When computing $\Loss(C)$,
 it is sufficient to insert zero-cost edges between all terminals $R_C$ into $C$
 instead of considering the $C/R_C$ contraction,
 see \citep[Lemma 2]{RZ05}.

The idea behind the LCA in contrast to the GCF
 is to leave out high-cost edges
 as long as they are not necessary to connect the solution.
From another point of view,
 this allows the algorithm
 to reject edges in a full component after a full component has already been accepted for inclusion.

In order to be able to perform a loss-contraction,
 the full component has to be included into $M$ first.
Hence $M/\Loss(C)$ is short-hand for $(M \cup C)/\Loss(C)$.
Since only the MSTs of the contracted graphs are necessary,
 we can perform a contraction by adding zero-cost edges between contracted terminals,
 and a loss-contraction by adding the non-loss edges.
Figure~\ref{fig:contraction-vs-loss-contraction}
 illustrates the difference of a contraction and a loss-contraction.

\begin{figure}
 \centering
 \subfigure[$\MST{M}$]{%
  \includegraphics{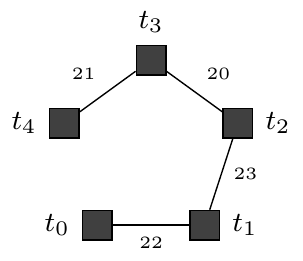}
 }
 \subfigure[Full component $C$ to be selected. Dashed lines denote loss edges.]{%
  \includegraphics{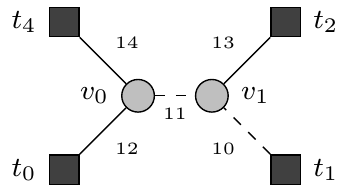}
 }\\
 \subfigure[$\MST{M/C}$]{%
  \includegraphics{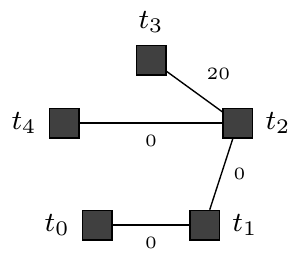}
 }
 \subfigure[$\MST{M/\Loss(C)}$]{%
  \includegraphics{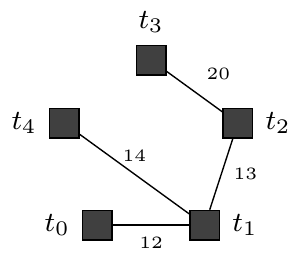}
 }
 \caption{An example showing the difference between a contraction (using zero-cost edges) and a loss-contraction of a full component.}
 \label{fig:contraction-vs-loss-contraction}
\end{figure}

Among the purely combinatorial algorithms,
 the GCF and its variant LCA are the ones
 we will focus on.
We refrain from explicitly implementing
 the mentioned GCF variants
 involving iterative and preprocessing techniques
 as they are either impractical~\citep{HP99}
 or dominated by other methods~\citep{BR94,KZ95}.
Also the algebraic approach~\citep{PS97} for the $3$-restricted case
 is clearly impractical.

\paragraph{Win functions.}
One crucial ingredient of many win functions is the \emph{save}.
When inserting a full component $C$ into $\MST{M}$,
 we obtain cycles.
Those cycles can be broken by deleting the maximum-cost edges in these cycles.
We call those edges \emph{save edges},
 and their total cost is the \emph{save}.
Formally,
 let $\save_M(u,v)$ be the maximum-cost edge
 on the unique path between $u,v \in \MST{M}$,
 and
 let
 $\save(M,C) := d(\MST{M}) - d(\MST{M/C})$
 denote the cost difference
 between the minimum spanning trees
 in $M$ and $M/C$.

Several win functions have been proposed.
Zelikovsky
 originally suggested the \emph{absolute} win function
  $\win_\text{abs}(M,C) := \save(M,C) - d(C)$
  that describes the actual cost reduction of $M$
  when we include $C$.
Using $\win_\text{abs}$ yields an $11/6$-approximation
 for $k=3$~\citep{Z92,Z93:IPL,Z93:AI}.
The \emph{relative} win function
 $\win_\text{rel}(M,C) := \frac{\save(M,C)}{d(C)}$
 achieves approximation ratio $1.69$
  for $k \to \infty$~\citep{Z95}.

For the LCA,
 \citet{RZ05}
 proposed
 $\win_\text{loss}(M,C) := \frac{\win_\text{abs}(M,C)}{d(\Loss(C))}$.
It relates the cost reduction from the choice of a full component
 to the cost of connecting their nonterminals to terminals,
 which is the actual cost when contracting the loss of a full component.
In their survey,
 \citet{GHNP01}
 used
 $\win_{\text{loss}'}(M,C) := \frac{\save(M,\Loss(C))}{d(\Loss(C))}$
 that coincides with $\win_\text{loss}(M,C) + 1$
  since $d(\MST{M/\Loss(C)}) = d(\MST{M/C}) + d(C) - d(\Loss(C))$.
We see that $\win_\text{loss}$ is conceptually a direct transfer of $\win_\text{rel}$ to the case of loss contractions.
It guarantees an $\approx 1.549$ approximation ratio for $k \to \infty$.

The GCF loop terminates when no promising full component has been found,
 that is,
 when the choice of a full component $C$ with maximum win does not improve $M$.
This is the case if
 $\win_\text{abs}(M,C) \leq 0$ (and hence $\win_\text{loss}(M,C) \leq 0$)
 or
 $\win_\text{rel}(M,C) \leq 1$
  for all $C \in \mathcal{C}_k$.

\subsection{Algorithms Based on Linear Programming}
\label{sec:LP}

In contrast to the purely combinatorial algorithms above,
 there are also approximation algorithms
 based on linear programming.

The primal-dual algorithm by \citet{GW95}
 for constrained forest problems can be applied
 to the STP
 but only yields a $2$-approximation.
It is based on the \emph{undirected cut relaxation}~(UCR)
 with a tight integrality gap of $2$.
We obtain the \emph{bidirected cut relaxation}~(BCR)
 by transforming $G$ into a bidirected graph.
Let $A := \setwith{\arc{u}{v}, \arc{v}{u}}{e = \edge{u}{v} \in E}$ denote the arc set of $G$,
 and let $\inarc{U} := \setwith{\arc{u}{v} \in A}{u \in U\setminus{V}, v \in U}$
  be the set of arcs entering $U \subseteq V$.
The cost of each arc coincides with the cost of the corresponding edge.
BCR is defined as:
  \linprog{min}{\tag{BCR}\sum_{e \in A}{d(e)\,x_e}}{%
   \constraint{\sum_{e \in \inarc{U}}{x_e}}{\geq}{1}{for all $U \subsetneq V$ with $r \in U \cap R \neq R$,}
    \label{eq:bcr:dcut} \\
   & 0 \leq x_e \leq 1\text{,} & & \text{for all $e \in A$}
  }
 where $r \in R$ is an arbitrary (fixed) \emph{root} terminal.
We obtain the ILP of BCR by requiring integrality for $x$
 (and analogously for the relaxations below).
Clearly, every feasible solution of the ILP
 spans all terminals:
 the \emph{directed cut constraint}~\eqref{eq:bcr:dcut}
  guarantees that there is at least one directed path from any terminal to $r$.
Since every optimal solution represents a tree where all arcs are directed towards $r$,
 dropping the directions of that arborescence yields a minimum-cost Steiner tree.
Although BCR is strictly stronger than UCR,
 no BCR-based approximation with ratio smaller than $2$ is known.

\citet{BGRS13}
 incorporate the idea of using $k$-restricted components
 to find the \emph{directed-component cut relaxation}~(DCR).
They prove an upper bound of the integrality gap of $1 + \frac{\ln{3}}{2} \approx 1.55$ for $k \to \infty$
 and obtain an approximation algorithm with approximation ratio at most $\varrho_k \ln{4}$
 that tends to $\approx 1.39$ for $k \to \infty$.
Let $\mathcal{D}_k$ denote the set of directed full components
 obtained from $\mathcal{C}_k$:
  For each $C \in \mathcal{C}_k$ with $R_C = \set{v_1, \ldots, v_{|R_C|}}$,
   we make $|R_C|$ copies of $C$
   and direct all edges in the $i$-th copy of $C$
  towards $v_i$
  for $i = 1, \ldots, |R_C|$.
For each $D \in \mathcal{D}_k$ let $t_D$ be the node all edges are directed to.
Let $\inarcX{U} := \setwith{D \in \mathcal{D}_k}{U \setminus R_D \neq \emptyset, t_D \in U}$ be the set of directed full components
 entering $U \subseteq R$.
DCR is defined as:
 \linprog{min}{\tag{DCR}\sum_{D \in \mathcal{D}_k}{d(D) \, x_D}}{%
  \constraint{\sum_{D \in \inarcX{U}}{x_D}}{\geq}{1}{for all $U \subsetneq R$ with $r \in U$,}
   \label{eq:dcr:cut} \\
   & 0 \leq x_D \leq 1\text{,} & & \text{for all $D \in \mathcal{D}_k$}
 }
 where $r \in R$ is again an arbitrary root.
The approximation algorithm
 iteratively solves DCR,
 samples a full component $D$ according to a probability distribution based on the solution vector,
 contracts $D$,
 and iterates this process by resolving the new DCR instance.
The algorithm stops when all terminals are contracted.
The union of the chosen full components represents the resulting $k$-restricted Steiner tree.
Although the sampling of the full components is originally randomized,
 a derandomization of the algorithm is possible.

\citet{W98} showed that constructing a minimum $k$-restricted Steiner tree
 is equivalent of finding
 a minimum spanning tree in the hypergraph
 $(R, \setwith{R_C}{C \in \mathcal{C}_k})$,
 \ie
 the terminals represent the nodes of the hypergraph
 and the full components represent the hyperedges.
He introduced the following relaxation:
 \linprog{min}{\tag{SER}\sum_{C \in \mathcal{C}_k}{d(C) \, x_C}}{%
  \constraint{%
   \sum_{C \in \mathcal{C}_k}{(|R_C| - 1) \, x_C}
   }{=}{|R| - 1}{}
  \label{eq:ser:full} \\
  \constraint{%
   \sum_{\substack{%
     C \in \mathcal{C}_k \\
     R' \cap R_C \neq \emptyset
    }}{(|R' \cap R_C| - 1) \, x_C}
   }{\leq}{|R'| - 1}{for all $R' \subseteq R, |R'| \geq 2$,}
  \label{eq:ser:subset} \\
   & 0 \leq x_C \leq 1\text{,} & & \text{for all $C \in \mathcal{C}_k$.}
}
Constraint~\eqref{eq:ser:full}
 represents the basic relation between the number of nodes and edges in hypertrees
 (like $|E| = |V|-1$ in trees).
Since arbitrary subsets of nodes are not necessarily connected but cycle-free,
 this directly implies the \emph{subtour elimination constraints}~\eqref{eq:ser:subset}.
We call that relaxation the \emph{subtour elimination relaxation}~(SER).

\citet{PV03}
 proved that DCR and SER
 are equivalent.
\citet{KPT11}
 and \citet{CKP10:IPCO} also provided partition-based relaxations
 that are equivalent to DCR and SER.
All these equivalent LP relaxations are summarized as \emph{hypergraphic relaxations}.

\citet{CKP10}
 developed an $1.55$-approximation algorithm for $k \to \infty$
 based on SER
 to prove the integrality gap of DCR in a simpler way than \citet{BGRS13}.
Their algorithm has the advantage that it only solves the LP relaxation once
 instead of solving new LP relaxations after each single contraction.
This improves the running times.
The disadvantage is that the approximation ratio is not better than
 the purely combinatorial algorithm by \citet{RZ05}.

\citet{GORZ12} used techniques from the theory of matroids and submodular functions
 to improve the upper bound on the integrality gap of the hypergraphic relaxations
 such that it matches the ratio $1.39$ of the approximation algorithm by \citet{BGRS13}.
They found a new approximation algorithm
 that solves the hypergraphic relaxation once
 and builds an auxiliary directed graph from the solution.
Full components in that auxiliary graph are carefully selected and contracted,
 until the auxiliary graph cannot be contracted any further.

We will focus on this latter algorithm
 and describe it in the remaining section.
Although the description of the algorithm by \citet{BGRS13} is quite simple,
 we have not chosen to implement it.
It is evident that it needs much more running time
 since the LP relaxation has to be re-solved in each iteration.
To this end,
 a lot of max-flows have to be computed on auxiliary graphs.
In contrast, the algorithm by \citet{GORZ12}
 only solves one LP relaxation
 and then computes some min-cost flows on a shrinking auxiliary graph.

\paragraph{Solving the LP relaxation.}
\label{sec:goemans-alg}

First, we have to solve the hypergraphic LP relaxation.
The number of constraints in both above relaxations
 grows exponentially with the number of terminals,
 but both relaxations can be solved in polynomial time
 using \emph{separation}:
  We first solve the LP for a subset of the constraints.
  Then, we
   solve the \emph{separation problem}, \ie search for some further violated constraints,
   add these constraints,
   resolve the LP,
   and iterate the process until there are no further violated constraints.
An LP relaxation with exponentially many constraints can be solved in polynomial time
 iff its separation problem can be solved in polynomial~time.

The separation problem of DCR includes a typical cut separation.
We generate an auxiliary directed graph
 with nodes $R$.
For each $D \in \mathcal{D}_k$,
 we insert
 one node $z_D$,
 an arc~$\arc{t_D, z_D}$
 and arcs~$\arc{z_D, w}$ for each $w \in R_D\setminus\set{t_D}$.
The inserted arcs are assigned capacities~$\bar{x}_D$
 where $\bar{x}$ is the current solution vector.
We can then check if
 there is a maximum flow from the chosen root $r$
 to a $t \in R\setminus\set{r}$
 with value less than $1$.
In that case,
 we have to add constraint~\eqref{eq:dcr:cut}
 with $U$ being a minimum cut set
  of nodes containing $r$.
Otherwise all necessary constraints have been generated.

A disadvantage of DCR over SER
 is that it has $k$ times more variables,
but
 cut constraints can usually be separated more efficiently
 than subtour elimination constraints.
However, \citet[App.~A]{GORZ12}
 provide a routine for SER that boils down to only max-flows,
 similar to what would be required for DCR as well.

First, we observe that
 \begin{alignat}{2}
  \constraint{\sum_{C \in \mathcal{C}_k\colon v \in C}{x_C}}{\geq}{1}{for all $v \in R$}
  \label{eq:ser:yv}
 \end{alignat}
 follows from projecting \eqref{eq:dcr:cut} onto $\mathbb{R}^{|\mathcal{C}_k|}$
  and by equivalence of DCR and SER.
We can start with the relaxation using only constraints~\eqref{eq:ser:full} and~\eqref{eq:ser:yv}.
Let $\bar{x}$ be the current fractional LP solution.
Let $\mathcal{\bar{C}}_k := \setwith{C \in \mathcal{C}_k}{\bar{x}_C > 0}$
 be the set of all (at least fractionally) chosen full components,
 and $y_r := \sum_{C \in \mathcal{C}_k\colon r \in R_C}{\bar{x}_C}$
 the `amount' of full components covering some $r \in R$.
We have $y_r \geq 1$ by~\eqref{eq:ser:yv},
 which is necessary for the separation algorithm to work correctly.

We construct an auxiliary network $N$ as follows.
We build a directed version of every chosen full component $C \in \mathcal{\bar{C}}_k$
 rooted at an arbitrary terminal $r_C \in R_C$.
The capacity of each arc in $C$ is simply $\bar{x}_C$.
We add a single source $s$ and arcs $(s, r_C)$ with capacity $\bar{x}_C$ for each $C$,
 as well as
 a single target $t$ and arcs $(r, t)$
 with capacity $y_r - 1$ for all $r \in R$.

For each $r \in R$,
 the separation algorithm
 computes a minimum $s$-$\set{r,t}$-cut in $N$.
Let $T$ be the node partition with $t \in T$
 and $\gamma$ the cut value.
Constraint~\eqref{eq:ser:subset} is violated for $R' := R \cap T$
 iff $\gamma < \sum_{r \in R}{y_r} - |R| + 1$.
If no violated constraints are found, $\bar{x}$ is a feasible and optimal fractional solution to SER.

\paragraph{The algorithm by Goemans~et~al.\ [2012].}%
Based on an optimal fractional solution to SER,
 the algorithm constructs an integral solution
 with an objective value that is at most $\varrho_k \ln{4}$ times worse than of the fractional solution.
This results in an approximation ratio and integrality gap of at most $\varrho_k \ln{4}$.
The algorithm has a randomized behavior
 but can be derandomized using dynamic programming
 (further increasing the running time by $\bigO{|V_C|^k}$ for each $C \in \mathcal{\bar{C}}_k$).
We focus on the former variant
 where the approximation ratio is not guaranteed but expected.

Let $\bar{x}$ be an optimal fractional solution to SER.
Initially, the algorithm constructs the auxiliary network $N$ representing $\bar{x}$
 as discussed for the separation.
Let $\mathcal{C}_N$ be the set of all components in $N$.
For each $C \in \mathcal{C}_N$, a set of core edges (see Def.~\ref{def:loss-core}) is computed.
Random core edges are sufficient for the expected approximation ratio.\footnote{%
  The derandomization performs this selection via dynamic programming.
  In contrast to \citet{BGRS13},
   the actual component selection (see below)
   is not randomized.
 }
In $N$, we add an arc $(s,v)$ for each core edge $e = (u,v)$,
 with the same capacity as for $e$.
In the main loop,
 we select beneficial components of $\mathcal{C}_N$ to contract,
 and modify $N$ to represent a \emph{feasible} solution
 for the contracted problem.
This is repeated until all components are contracted.
The contracted full components form a $k$-restricted Steiner tree.

The nontrivial issue here is to guarantee feasibility of the modified network.
Contracting the selected $C \in \mathcal{C}_N$ would make $N$ %
 infeasible.
It suffices to remove some core edges to reestablish feasibility.
The \emph{minimal} set of core edges that has to be removed
 is a set of bases of a matroid,
 and can hence be found in polynomial time.
For brevity, we call a basis of such a matroid for the contraction of $C$
 the \emph{basis for $C$}.

In each iteration, the algorithm selects a suitable full component $C$
 and a basis for $C$ of maximum weight.
However, the weight of a basis is not simply its total edge cost:
 After removing core edges,
 there are further edges that can be removed without affecting feasibility
 and whose costs are incorporated in the weight of the basis.
Computing the maximum-weight basis for $C$
 boils down to a min-cost flow computation.

\section{Algorithm Engineering}
\label{sec:alg-eng}

We now have a look at different algorithmic variants of the strong algorithms
 to achieve improvements for the practical implementation.
All variants do not affect the asymptotical runtime
 but may be beneficial in practice.
Since all strong algorithms are based on full components,
 we will look at the construction of full components first.
Afterwards we look at the concrete algorithms,
 the GCF/LCA
 and the algorithm based on SER.

\subsection{Generation of Full Components}

We consider three ways to generate the set $\mathcal{C}_k$ of $k$-restricted components.

The first one is the enumeration of full components,
 that is, for a given subset of terminals $R', |R'| \leq k$,
 we construct every full component on $R'$
 and check which one has minimum cost.
We call this strategy \texttt{gen=all}.

The second one is the generation using Voronoi regions.
This differs from the enumeration method
 in that we only test full components for optimality
 where the inner nodes lie in Voronoi regions of the terminals.
We call this strategy \texttt{gen=voronoi}.

These two generation strategies generate $\mathcal{C}_k$
 in a first phase.
The actual approximation algorithms then simply iterate over (and possibly delete from) this precomputed set
$\mathcal{C}_k$.

In contrast,
 the third strategy generates a full component
 when it is needed in the Greedy Contraction Framework.
Hence, $\mathcal{C}_k$ is only used implicitly.
This strategy is called \texttt{gen=ondemand}.

Before we discuss these variants and their applicability
 in more detail,
 we consider different strategies for the computation of shortest paths.

\paragraph{Precomputing shortest paths.}
For each of the above mentioned methods to construct full components,
 we need a fast way to retrieve shortest paths from any node to another node
 very often.
We may achieve this efficiently by
 precomputing an all-pairs shortest paths~(APSP) lookup table
  in time $\bigO{|V|^3}$ once,
 and then looking up predecessors and distances in $\bigO{1}$.

For $k=3$, there is at most one nonterminal with degree $3$ in each full component.
Hence, to build $3$-components,
 we only need shortest paths for pairs of nodes where at least one node is a terminal.
This allows us to only compute the single-source shortest paths~(SSSP) from each terminal
 in time $\bigO{|R| \cdot |V|^2}$.

We call the two above strategies \texttt{dist=apsp}
 and  \texttt{dist=sssp}, respectively.

We observe that since a full component must not contain an inner terminal,
 we need to obtain shortest paths over nonterminals only. %
We call such a shortest path \emph{valid}.
This allows us to rule out full components before they are generated.

We can modify both the APSP and SSSP computations
 such that they never find paths over terminals.
This way, the running time decreases when the number of terminals increases.
The disadvantage is that paths
 with detours over nonterminals
 are obtained.
We call this strategy \texttt{sp=forbid}.

Another way is to modify the APSP and SSSP computations
 such that they \emph{prefer} paths over terminals in case of a tie,
 and afterwards removing such paths
  (invalidating certain full components altogether).
That way we expect to obtain fewer valid shortest paths,
 especially in instances where ties are common,
 for example,
 instances from VLSI design or complete instances.
This strategy is called \texttt{sp=prefer}.

\paragraph{Enumeration of full components.}
For arbitrary $k$,
 the enumeration of full components is the only method
 to generate the list $\mathcal{C}_k$.
Note that for any $U \subseteq V$,
 $\bar{G}_U$ can be constructed in $\bigO{|U|^2}$
 by a lookup in the distance matrix for each node pair of~$U$.
A naïve construction of $\mathcal{C}_k$, given in~\citep{RZ05}, is as follows:
 for each subset $R' \subseteq R$ with $2 \leq |R'| \leq k$,
  compute $M_{R'}$
   as the smallest $\MST{\bar{G}_{R' \cup S}}$
   over all subsets $S \subseteq V\setminus{R}$ with $|S| \leq |R'| - 2$.
 We insert $M_{R'}$
  into $\mathcal{C}_k$
  if it does not contain inner terminals.
We call this strategy \texttt{gen=all:naïve}.
The time complexity (considering $k$ as input)
 is
 $
  \bigO{|R|^k |V \setminus R|^{k-2} k^4}
 $.

Reflecting on this procedure,
 we can do better.
The above approach requires many MST computations.
We can save time by precomputing a list $\mathcal{L}$ of
 potential inner trees of full components, that is,
 we store trees without any terminals.
For all applicable terminal subset cardinalities $t = 2, 3, \ldots, k$,
 we perform the following step:
\begin{enumerate}
 \item
  For all subsets $S \subseteq V \setminus R$ with $|S| = t-2$,
  we insert $\MST{\bar{G}_S}$ into $\mathcal{L}$, and
 \item
  for all subsets $R' \subseteq R$ with $|R'| = t$,
   we iterate over all trees $T$ in $\mathcal{L}$,
    connect each terminal in $R'$ to $T$
     as cheaply as possible,
    and insert the minimum-cost tree (among the constructed ones) into $\mathcal{C}_k$.
\end{enumerate}
We denote this strategy by \texttt{gen=all:smart}.
Its time complexity (again considering $k$ as input) is
 $
  \bigO{( |V\setminus R|^{k-2} + k |R|^k ) \, k^3 }
 $.

One disadvantage
 is that components may be generated that will never be used:
 in step 2,
 it can happen that the path from a terminal in $R'$ to another terminal in $R'$
  is cheaper than the cheapest path from a terminal in $R'$ to a nonterminal in $T$ in $\mathcal{L}$.
In this case,
 \texttt{gen=all:naïve} would not insert
 the constructed tree into $\mathcal{C}_k$
 since it contains an inner terminal
 (and could hence be decomposed into full components $C_1$ and $C_2$
  that are already included).
For \texttt{gen=all:smart}, a full component $C$ with larger cost is inserted into $\mathcal{C}_k$.
However, $C$ is never used in any of the algorithms since $d(C_1) + d(C_2) \leq d(C)$.

Although these general constructions work for all values of $k$,
 it is useful for the actual running time to
 make some observations for small values:
 $2$-components are exactly the shortest paths between any pair of terminals,
  so a $2$-component is essentially computed by a lookup.
 Moreover,
  $2$-components are not used in GCF and LCA, so these lookups can be skipped.
 For $3$-components, the graphs in $\mathcal{L}$ are single nodes.
 Generating $\mathcal{L}$ can hence be omitted and we directly iterate over all nonterminals instead.
We apply these observations for \texttt{gen=all:smart} and \texttt{gen=all:naïve}.

\citet{CGSW14} propose to compute $\mathcal{C}_k$
 essentially by running the first $k$ iterations of the dynamic programming algorithm by \citet{DW72}.
It is based on the simple observation that
 in order to compute a minimum-cost tree spanning $k$ terminals,
  the Dreyfus-Wagner algorithm also computes all minimum-cost trees spanning \emph{less than} $k$ terminals.
Hence one call to this restricted Dreyfus-Wagner algorithm yields all $k$-restricted components.
This is especially interesting for $k > 3$.
We denote this method by \texttt{gen=all:dw}.
The time complexity
 is
 $
  \bigO{|R|^{k} |V| (2^k + |V|) \, k}
 $.

\paragraph{Using Voronoi regions to build full components.}
\citet{Z93:AI}
 proposed to use Voronoi regions to obtain a faster algorithm for full component construction.
A \emph{Voronoi region}
  $\Vor{r} = \set{r} \cup \setwith{v \in V\setminus{R}}{d(r,v) \leq d(s, v) \;\forall s \in R}$
 of a terminal $r \in R$
 is the set of nodes that are nearer to $r$
 than to any other terminal.
Since we want the set of Voronoi regions of each terminal to be a partition of $V$,
 a node $v$ with $d(r, v) = d(s, v)$, $s \neq r$,
  is arbitrarily assigned either to $\Vor{r}$ or to $\Vor{s}$.
Voronoi regions can be computed efficiently using one multi-source shortest path computation
 where the terminals are the sources.
This can be performed using a trivially modified Dijkstra shortest path algorithm,
 or by adding a super-source,
  connecting it to all terminals with zero distance,
  and applying a single-source shortest path algorithm from the super-source~\citep{M88}.

The basic idea of the full component construction using Voronoi regions
 is as follows:
 when we want to construct a minimum full component on terminals $R'$
 we only consider the nonterminals in $\Vor{R'} := \bigcup_{t \in R'}\Vor{t}$
 instead of all nonterminals in $V$.
Note that this is only a practical improvement to the naïve enumeration
 and does not affect the asymptotic worst-case behavior.

\begin{figure}[p]
 \centering
 \subfigure[Complete instance graph $G$. Invisible edges have distance costs.]{
  \includegraphics{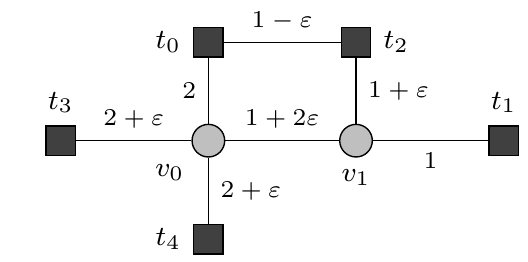}
  \label{fig:voronoi-counterexample:instance}
 }\\
 \subfigure[Minimum solution with cost $8 + 4\varepsilon$.]{
  \includegraphics{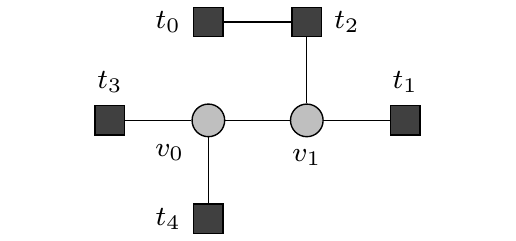}
  \label{fig:voronoi-counterexample:best}
 }
 \subfigure[Minimum Voronoi-based solution with cost $9 + 3\varepsilon$.]{
  \includegraphics{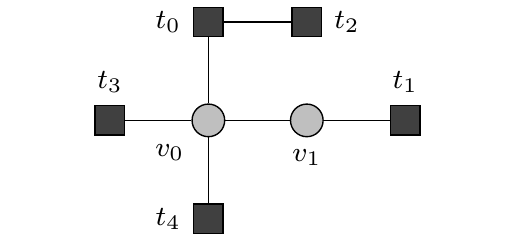}
  \label{fig:voronoi-counterexample:best-voronoi}
 }
 \subfigure[All $2$-components and their costs.]{
  \includegraphics{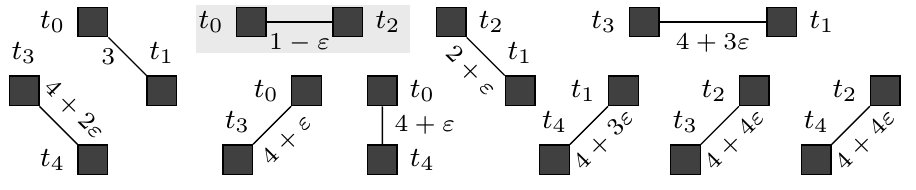}
  \label{fig:voronoi-counterexample:all-2-components}
 }
 \subfigure[All $4$-components that can be constructed with and without using Voronoi regions.]{
  \includegraphics{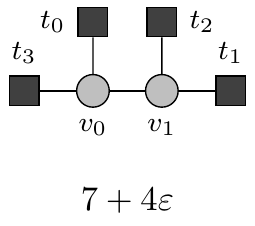}
  \quad
  \includegraphics{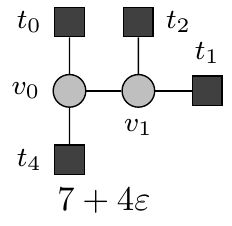}
  \quad
  \includegraphics{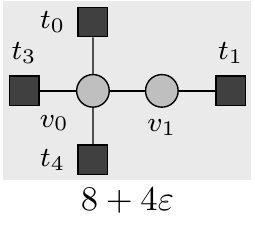}
  \label{fig:voronoi-counterexample:4-components-possible}
 }
 \subfigure[The $4$-components that cannot be constructed using Voronoi regions.]{
  \includegraphics{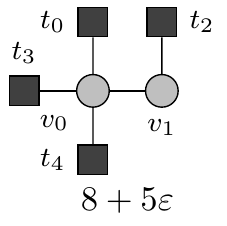}
  \quad
  \includegraphics{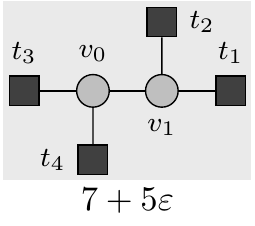}
  \label{fig:voronoi-counterexample:4-components-impossible}
 }
 \quad
 \subfigure[The latter $4$-components constructed using Voronoi regions.]{
  \includegraphics{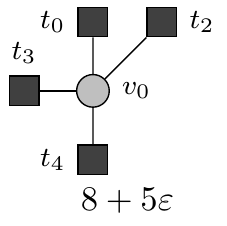}
  \quad
  \includegraphics{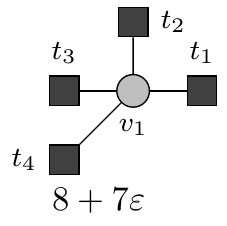}
  \label{fig:voronoi-counterexample:4-components-voronoi}
 }
 \caption{%
  Fig.~\ref{fig:voronoi-counterexample:instance} shows the complete instance $G$.
  The costs of the invisible edges coincide with the distances between each node pair.
  The unique Voronoi regions are
   $\Vor{t_i} = \set{t_i, v_i}$ for $i \in \set{0, 1}$
   and
   $\Vor{t_j} = \set{t_j}$ for $j \in \set{2,3,4}$.
  By the structure of the instance,
   a minimum $4$-restricted Steiner tree
   contains a $4$-component and a $2$-component,
   \ie $3$-components are not beneficial.
  Fig.~\ref{fig:voronoi-counterexample:best} shows the minimum $4$-restricted Steiner tree
   that is obtained by an enumeration of full components.
  Fig.~\ref{fig:voronoi-counterexample:best-voronoi} shows the minimum $4$-restricted Steiner tree
   that is obtained by full components that are constructed from Voronoi regions only.
  The tree in Fig.~\ref{fig:voronoi-counterexample:best} is constructed from the full components
   shown in Figures~\ref{fig:voronoi-counterexample:all-2-components},
    \ref{fig:voronoi-counterexample:4-components-possible}
    and~\ref{fig:voronoi-counterexample:4-components-impossible}
  whereas the tree in Fig.~\ref{fig:voronoi-counterexample:best-voronoi} uses Fig.~\ref{fig:voronoi-counterexample:4-components-voronoi}
   instead of~\ref{fig:voronoi-counterexample:4-components-impossible}.
  The full components used in Fig.~\ref{fig:voronoi-counterexample:best} and~\ref{fig:voronoi-counterexample:best-voronoi}
   are marked.
 }
 \label{fig:voronoi-counterexample}
\end{figure}

The question arises whether such a construction always leads to a set of full components
 that is necessary to obtain a minimum $k$-restricted Steiner tree.
\citet{Z93:AI} showed that GCF with $k=3$ and $\win_\text{abs}$
 finds the same maximum win in each iteration
 no matter if \texttt{gen=all} or \texttt{gen=voronoi} is used.
An analogous argumentation can be applied to prove this for $\win_\text{rel}$.
We generalize this result to see that Voronoi regions can always be used for $k=3$.
The following lemma shows that a minimum $3$-restricted Steiner tree can be obtained
 even if $\mathcal{C}_3$ is generated using Voronoi regions.

\begin{lemma}
 Let $T$ be a $3$-restricted Steiner tree for $\bar{G}$.
 There is a $3$-restricted Steiner tree $T^\ast$ for $\bar{G}$
  with $d(T^\ast) \leq d(T)$
  such that for each nonterminal $v_C$ in any full component $C$ in $T^\ast$,
  we have $v_C \in \Vor{R_C}$.
\end{lemma}
\begin{proof}
 As $T$ is $3$-restricted,
  we can w.l.o.g.\ assume that
  no nonterminal in $T$ is adjacent to another nonterminal.
 Hence there is at most one nonterminal in each full component.
 On the other hand,
  any $3$-component $C$ has at least one nonterminal.
  Let $v_C$ be the unique nonterminal in $C$.
 If $v_C \in \Vor{R_C}$ for each $3$-component $C$,
  we are done.

 Let $C$ be a $3$-component with $v_C \notin \Vor{R_C}$.
 There must be a terminal $s \in R \setminus R_C$ such that $v_C \in \Vor{s}$.
 Consider the three connected components that emerge by removing $v_C$ from $T$.
 Let $t$ be the unique terminal adjacent to $v_C$
  that is in the same connected component as $s$.
 Replacing edge $\edge{v_C}{t}$ by $\edge{v_C}{s}$
  in the original $T$
  results in a tree $T'$ with $d(T') \leq d(T)$.
 We obtain $T^\ast$ by repeating this process. %
\end{proof}
The proof does not generalize to $k \geq 4$:
 let $u,v \in V_C\setminus{R_C}, v \notin \Vor{R_C}$;
 replacing an edge can remove $\edge{v}{t}$
 and insert $\edge{u}{s}$ instead of $\edge{v}{s}$.
That way, it is possible that the cost of the resulting component increases.
For example, the $4$-component of Fig.~\ref{fig:voronoi-counterexample:best}
 would become the $4$-component of Fig.~\ref{fig:voronoi-counterexample:best-voronoi}.
Fig.~\ref{fig:voronoi-counterexample}
 leads to the following observation.
\begin{observation}
For $k \geq 4$,
 the Voronoi-based approach does not guarantee
 minimum $k$-restricted Steiner trees.
Even for $k = 4$,
 one can obtain solutions at least $9/8$ times worse than using full component enumeration.
\end{observation}

\paragraph{Direct generation of $3$-restricted components (\texttt{gen=ondemand}).}

This way of generating full components has been proposed by \citet{Z93:IPL}.
It is only available for $3$-restricted components and $\win_\text{abs}$ in the Greedy Contraction Framework.

In contrast to the other generation strategies, there is no explicit generation phase.
An optimal full component is directly constructed when necessary.
Therefore, we iterate over all nonterminals $v \in V \setminus R$.
In each iteration, we
 \begin{enumerate}
  \item find $s_0 \in R$
   with minimum distance $d(v, s_0)$ to $v$,
  \item find $s_1 \in R\setminus\set{s_0}$
   with maximum $d(\save_M(s_0, s_1)) - d(v, s_1)$,
  \item find $s_2 \in R\setminus\set{s_0,s_1}$
   with maximum $\win_\text{abs}(M,C) = \save(M,C) - d(v, s_0) - d(v, s_1) - d(v, s_2)$
   where $C$ is the full component of terminals $s_0, s_1, s_2$ with center $v$,
 \end{enumerate}
 and keep the full component $C$ with maximum $\win_\text{abs}(M,C)$.

\subsection{Greedy Contraction Framework}

\paragraph{Reduction of the full component set.}
We can show that
 the win %
 of any $C \in \mathcal{C}_k$
 will never increase
 during the execution of GCF or LCA.
This follows from
 $\save(M, C_2) \geq \save(M/C_1, C_2)$
 and
 $\save(M, C_2) \geq \save(M/\Loss(C_1), C_2)$
 for any full components $C_1, C_2$,
 which we prove in the following lemma.

\begin{lemma}
 Consider the metric complete graph $M$
  in any iteration of GCF or LCA.
 Let $u,v \in R$ with $u \neq v$,
  edge $e = \edge{u}{v}$ with arbitrary cost $d(e)$,
  and $M_e := (V_M, E_M \cup \set{e})$.
 We have
  $\save(M, C) \geq \save(M_e, C)$.
\end{lemma}
\begin{proof}
 Consider $\MST{M}$.
 Inserting edge $e$ closes a cycle,
  so either $e$ or $f := \save_M(u,v)$ will not be in $\MST{M_e}$.
 We have $d(\MST{M_{e}}) = d(\MST{M}) + \min\set{0, d(e) - d(f)}$,
  and thus
 \begin{align*}
  \lefteqn{\save(M,C) - \save(M_e,C)} \\\qquad
   & = d(\MST{M}) - d(\MST{M/C}) - d(\MST{M_{e}}) + d(\MST{M_{e}/C}) \\
   & = d(\MST{M_{e}/C}) - d(\MST{M/C}) - \min\set{0, d(e) - d(f)}
 \end{align*}
  which proves the claim if $d(\MST{M_{e}/C}) \geq d(\MST{M/C})$.

 If $d(e) > d(f)$, we have $\MST{M_e} = \MST{M}$.
 Now consider $d(e) \leq d(f)$.
 For any $x,y \in R_C$
  with $\save_M(x,y) = f$,
  we have $d(\save_{M_e}(x,y)) \leq d(f)$.
 In any case, we have $d(\MST{M_{e}/C}) \geq d(\MST{M/C})$
  and thus the claim holds.
\end{proof}

We can utilize this fact to reduce the number
 of full components.
Every time we find a non-promising full component,
 we remove that full component from $\mathcal{C}_k$.
In particular,
 when we construct a non-promising full component,
 we discard it already
 before inserting it into $\mathcal{C}_k$.
We denote this variant by \texttt{reduce=on}.
Note that this variant does not work together with \texttt{gen=ondemand}.

\paragraph{Save computation.}

In order to compute the win of a full component $C$, %
 we first have to compute
 $\save(M,C) = d(\MST{M}) - d(\MST{M/C})$.
Doing a contraction and MST computation
 for each potential component
 would be cumbersome and inefficient.

Since we can consider a contraction of $u$ and $v$
 as an insertion of zero-cost edges $\edge{u}{v}$,
 we can construct $\MST{M/C}$ from $\MST{M}$
 by removing $\save_M(u,v)$ and inserting a zero-cost edge $\edge{u}{v}$
  for each pair $u,v \in R_C$.
It follows that $\save(M,C)$ coincides with the total cost of the removed save edges.
If we are able to compute $\save_M(u,v)$ in, say, constant time,
 we are also able to compute $\save(M,C)$ in $\bigO{k}$.

One simple idea (also proposed by \citet{Z93:AI})
 is to build and use an $|R| \times |R|$ matrix
 to simply lookup the most expensive edges between each pair of terminals directly.
After each change of $M$,
 this matrix is (re-)built in time $O(|R|^2)$.
We call this method \texttt{save=matrix}.

The build times of the former approach can be rather expensive.
\citet{Z93:IPL} provided another approach
 that builds an auxiliary binary arborescence $W(T)$
 for a given tree $T := \MST{M}$.
The idea of $W(T)$ is to represent a cost hierarchy
 to find a save edge using lowest common ancestor queries.
We define $W(T)$ inductively:
 \begin{itemize}
  \item If $T$ is only one node $v$, $W(T)$ is a single node representing $v$.
  \item If $T$ is a tree with at least one edge,
   the root node $r$ of $W(T)$ represents the maximum-cost edge $e$ of $T$.
   By removing $e$,
    $T$ decomposes into two trees $T_1$ and $T_2$.
   The roots of $W(T_1)$ and $W(T_2)$ are the children of $r$ in $W(T)$.
 \end{itemize}
Nodes in $T$ are leaves in $W(T)$ and edges in $T$ are inner nodes in $W(T)$.
To construct $W(T)$, we first sort the edges by their costs
 and then build $W(T)$ bottom-up.
The construction time of $W(T)$, dominated by sorting, takes time $\bigO{|R| \log|R|}$.

We now want to perform lowest common ancestor queries on $W(T)$
 in $\bigO{1}$ time.
Let $n \in \bigO{|R|}$ be the number of nodes in $W(T)$.
Some preprocessing is necessary to achieve that.
The theoretically best known algorithm by \citet{HT84}
 needs time $\bigO{n}$ for preprocessing
 but is too complicated and cumbersome to implement and use in practice.
We hence use a simpler and more practical $\bigO{n \log n}$-algorithm by \citet{BF00}.
In either case, the time to build $W(T)$ and do the preprocessing is $\bigO{|R| \log|R|}$.

Instead of rebuilding $W(T)$ from scratch after a contraction,
 we can directly update $W(T)$ in time proportional to the height of $W(T)$.
This can be accomplished
 by adding a zero-cost edge-representing node $u_0$,
 moving the contracted nodes $u_1, u_2$ to be children of $u_0$,
 and then fixing $W(T)$ bottom-up from the former parents of $u_1$ and $u_2$
 up to the root.
On the way up, we remove the node that represents edge $\save_M(u_1, u_2)$
 as soon as we see it.

We denote the variant of fully rebuilding $W(T)$ by \texttt{save=static},
 and the variant of updating $W(T)$ by \texttt{save=dynamic}.
In the latter case,
 it is sufficient to update $W(T)$ only,
 without necessity to store $M$ or $T$ explicitly.

\paragraph{Evaluation passes.}

The original GCF needs $\bigO{|R|}$ passes in the evaluation phase.
In each iteration, the whole
  (probably reduced)
  list $\mathcal{C}_k$ of full components
 has to be evaluated.
We investigate another heuristic strategy
 that performs one pass
 and hence evaluates the win function of each full component
 at most twice.
The idea is to
 sort $\mathcal{C}_k$ in decreasing order by their initial win values.
Then we do one single pass over the sorted $\mathcal{C}_k$
 and contract the promising full components.
We call this strategy \texttt{singlepass=on}.

\begin{lemma}
 GCF with \texttt{singlepass=on} on $k$-restricted components
  has an approximation ratio smaller than $2$.
\end{lemma}
\begin{proof}
 Let $T$ be the Steiner tree solution of the MST-based $2$-approximation~\citep{KMB81}
  and $T^\ast$ be a minimum Steiner tree.
 Consider the case that there are promising $k$-restricted components.
 At least one of them, say $C$, will be chosen and contracted in the single pass.
 This yields a Steiner tree $T'$ with
  $d(T')
   \leq d(T) - \win_\text{abs}(T, C)
   < 2 d(T^\ast)
  $.
 Now consider the case that there are no promising $k$-restricted components.
 The algorithm would not find any full component to contract,
  but GCF with \texttt{singlepass=off} would also not contract any full component.
 By the approximation ratio of GCF, we have $d(T) < 2 d(T^\ast)$.
\end{proof}

\subsection{Algorithms Based on Linear Programming}

We now consider variants for the LP-based algorithm.

\paragraph{Solving the LP relaxation.}

First, we can observe that
 during separation, each full component's inner structure is irrelevant
 for the max-flow computation.
It hence suffices to insert a directed star into $N$ for each chosen full component.
That way, the size of $N$ becomes independent of $|V|$;
 this should hence be particularly beneficial
 if $|R|$ is small compared to $|V|$.

To solve SER,
 we start with constraints~\eqref{eq:ser:full}. %
Since we need~\eqref{eq:ser:yv} for the separation algorithm,
 we may include them in the initial LP formulation (denoted by \texttt{presep=initial})
 or add them iteratively when needed (\texttt{presep=ondemand}).

In the beginning of the separation process,
 it is likely
 that the hypergraph $(R, \mathcal{\bar{C}}_k)$ for a current solution $\bar{x}$ is not connected.
Hence it may be beneficial
 to apply a simpler separation strategy first:
  perform a connectivity tests and
  add~\eqref{eq:ser:subset} for each full component.
This variant is denoted by \texttt{consep=on}.

\paragraph{Pruning full leaf components.}

After solving the LP relaxation,
 the actual approximation algorithm
 with multiple minimum-cost flow computations
 in a changing auxiliary network
 starts.
However, solution $\bar{x}$
 is not always fractional.
If the solution \emph{is} fractional,
 there may still be full components $C$
 with $\bar{x}_C = 1$.
We will show that we can directly choose some of these integral components
 for our final Steiner tree
 and then generate a smaller network $N$ that does not contain them.
\begin{lemma}
 Let $C^\ast \in \mathcal{\bar{C}}_k$
  and $|R_{C^\ast} \cap \bigcup_{C \in \mathcal{\bar{C}}_k\setminus\set{C^\ast}}{R_{C}}| = 1$.
 Let $v^\ast$ be that one terminal.
 The solution $\bar{x}$
  obtained by setting $\bar{x}_{C^\ast} := 0$
  is feasible
  for the same instance with reduced terminal set $R \setminus (R_{C^\ast} \setminus \set{v^\ast})$.
\end{lemma}
\begin{proof}
 By constraint~\eqref{eq:ser:yv}
  we have $\bar{x}_{C^\ast} = 1$.
 We set $\bar{x}_{C^\ast} := 0$
  and $R := R \setminus (R_{C^\ast} \setminus \set{v^\ast})$,
  and observe how the left- (LHS) and the right-hand side (RHS)
   of the SER constraints change.
 The LHS of constraint~\eqref{eq:ser:full}
  is decreased by $(|R_{C^\ast}| - 1) x_{C^\ast} = |R_{C^\ast}| - 1$,
  its RHS is decreased by $|R_{C^\ast} \setminus \set{v^\ast}| = |R_{C^\ast}|-1$;
 constraint~\eqref{eq:ser:full} still holds.
 Consider~\eqref{eq:ser:subset}.
 On the LHS,
  the $\bar{x}_C$-coefficients of full components $C \neq C^\ast$ with $\bar{x}_C > 0$ are not affected
  since $R' \cap R_C$ contains no terminal from $R_{C^\ast} \setminus \set{v^\ast}$.
 The LHS hence changes by $\max(|R' \cap R_{C^\ast}| - 1, 0)$.
 The RHS is decreased by
  $|R' \cap (R_{C^\ast} \setminus \set{v^\ast})|$
   which coincides with the LHS change.
\end{proof}
Hence, we can always choose and contract such full leaf components
 without removing any core edges from outside that component
 and without expensive search.
We call this strategy \texttt{prune=on}.

\paragraph{Solving a stronger LP relaxation.}

One way that could help to improve the solution quality
 is to use a strictly stronger relaxation than SER.
Consider the constraints
 \begin{equation}
  \sum_{C \in \mathcal{C}'}{x_{C}} \;\leq\; S(\mathcal{C}')
  \qquad
   \forall
    \mathcal{C}' \subseteq \mathcal{C}_k,
  \label{eq:ser:cycle-nphard}
 \end{equation}
 where
  $S(\mathcal{C}')$
  is the maximum number of full components of $\mathcal{C}'$
   that can simultaneously be in a valid solution.
$S(\mathcal{C}')$ coincides with the maximum number of hyperedges
 that can form a subhyperforest in $H = (R, \setwith{R_C}{C \in \mathcal{C}'})$.
Unfortunately,
 obtaining $S(\mathcal{C}')$ is an NP-hard problem
 as can be shown by an easy reduction from \textsc{Independent Set}.

We try to solve the problem for a special case of $\mathcal{C}'$ only:
  \begin{alignat}{2}
   \sum_{C \in \mathcal{C}'}{x_{C}} \; & \leq \; 1
   & \qquad &
    \forall
     \mathcal{C}' \subseteq \mathcal{C}_k\colon
      |C_i \cap C_j| \geq 2
     \; \forall C_i, C_j \in \mathcal{C}' \text{.}
   \label{eq:ser:cycle1}
  \end{alignat}

\begin{lemma}
 \label{lemma:strictly-stronger}
 SER with constraint~\eqref{eq:ser:cycle1} is strictly stronger than SER.
\end{lemma}
\begin{proof}
 Clearly, the new LP is at least as strong as SER
  since we only add constraints.
 Let $G = (V,E)$ be a graph
  with $V = \set{v_0, \ldots, v_k, t_0, \ldots, t_k}$
  and $E = \setwith{\edge{v_i}{t_i}, \edge{v_{i}}{v_{i+1}}}{i=0,\ldots,k}, v_{k+1} = v_0$,
  with terminals $R = \set{t_0, \ldots, t_k}$,
  cost $1$ for all edges incident to any terminal,
  and cost $0$ for all other edges.
 Each full component $C$ has exactly cost $|R_C|$.

 Consider the solution $\bar{x}$
  with
  $$
   \bar{x}_C = \begin{cases}
    \frac{k}{k^2-1} & \text{if $|R_C| = k$,} \\
    0 & \text{otherwise}
   \end{cases}
  $$
  that is feasible to SER
  by
  \begin{align*}
   \sum_{C \in \mathcal{C}_k}{(|R_C| - 1) \bar{x}_C} &
     \; = \; \sum_{C \in \mathcal{C}_k \colon |R_C| = k}{(k - 1) \frac{k}{k^2-1}}
  \\ &
     \; = \; (k + 1)(k - 1) \frac{k}{k^2-1}
     \; = \; k
     \; = \; |R| - 1 \text{,}
  \end{align*}
   and
   \eqref{eq:ser:subset}
   holds since for
    $|R'| = k$ we have $\frac{k}{k^2 - 1} \leq k - 1$,
    and for
    $|R'| < k$ we have $0 \leq |R'| - 1$.
 The objective value for $\bar{x}$ is
  $
   \sum_{C \in \mathcal{C}_k}{|R_C| \bar{x}_C}
    = (k+1) k \frac{k}{k^2-1}
    = \frac{k^3 + k^2}{k^2 - 1}
  $.

 Let $\bar{X}_\ell := \sum_{C \in \mathcal{C}_k\colon |R_C| = \ell}{\bar{x}_C}$.
 Note that for any solution $\bar{x}$,~\eqref{eq:ser:full} and the objective function can be written as
  $\sum_{\ell = 2}^{k}{(\ell - 1) \bar{X}_\ell} = k$
  and
  $\sum_{\ell = 2}^{k}{\ell \bar{X}_\ell}$, respectively.
 Assume we decrease $\bar{X}_k$ by some $\varepsilon > 0$.
 Since
  $\sum_{\ell = 2}^{k}{(\ell - 1) \bar{X}_\ell}
   = k - (k-1)\varepsilon
   < k = |R| - 1$,
  we would have to increase $\bar{X}_i$ by $\frac{k-1}{i-1} \varepsilon_i$
   for all $i \in \set{2, \ldots, k-1}$
   to become feasible for~\eqref{eq:ser:full} again.
 Here, $\varepsilon_2, \ldots, \varepsilon_{k-1} > 0$
  are chosen
  such that they sum up to $\varepsilon$.
 This increases the objective value by
  $\sum_{i = 2}^{k-1}{i \, \frac{k-1}{i-1} \varepsilon_i} - k \varepsilon$
  which is clearly minimized
   by setting $\varepsilon_{k-1} := \varepsilon$ and $\varepsilon_i := 0$ for $i < k-1$.
 The increase of the objective value is hence at least
  $(k-1) \frac{k-1}{k-2} \varepsilon - k \varepsilon
   = \frac{1}{k-2} \varepsilon > 0
  $.

 Since
  \begin{align*}
   \sum_{C \in \mathcal{C}_k}{\bar{x}_C} &
    \; = \; (k + 1)\frac{k}{k^2-1}
    \; = \; \frac{k}{k - 1}
    \; = \; 1 + \frac{1}{k - 1}
  \end{align*}
  violates~\eqref{eq:ser:cycle1},
  we add the constraint
   $
   \sum_{C \in \mathcal{C}_k\colon |R_C| = k}{x_C} \leq 1
   $.
 We thus have to decrease $\bar{X}_k$ by $\frac{1}{k-1}$
  to form a feasible solution,
  which increases the objective value.
\end{proof}

Finding a $\mathcal{C}'$
 is equivalent to
 finding a clique in the conflict graph $G' = (\mathcal{C}_k, \setwith{\edge{C_i}{C_j}}{C_i, C_j \in \mathcal{C}_k, |R_{C_i} \cap R_{C_j}| \geq 2})$.
Based on the proof of Lemma~\ref{lemma:strictly-stronger},
 we restrict ourselves to cliques with at most $k+1$ nodes.
Such cliques
 can be found in polynomial time for constant $k$,
 in order to separate the corresponding constraints.
Observe that
 $G'$ need only be constructed from $\mathcal{\bar{C}}_k$
 instead of $\mathcal{C}_k$.
We call this strategy \texttt{stronger=on}.

\paragraph{Bounding the LP relaxation.}
An idea to improve the running time for the LP
 is to initially compute a simple $2$-approximation
 and apply its solution value as an upper bound on the objective value.
If this bound is smaller than the pure LP solution,
 there is no feasible solution to the bounded LP
 and we simply take the $2$-approximation.
We call this strategy \texttt{bound=on}.

\section{Experimental Evaluation}
\label{sec:exp-eval}

In the following experimental evaluation,
 we use an Intel~Xeon E5-2430 v2, 2.50\,GHz
 running Debian 8.
The binaries are compiled in 64bit
 with g++~4.9.0 and \texttt{-O3} optimization flag.
All algorithms are implemented as part of the free C++ Open Graph Drawing Framework (OGDF),
 the used LP solver is CPLEX~12.6.
We evaluate our algorithms with the 1\,200 connected instances from the \SteinLib library~\citep{KMV00},
 the currently most widely used benchmark set for STP.

We say that an algorithm \emph{fails} for a specific instance
 if it exceeds one hour of computation time or needs more than 16 GB of memory.
Otherwise it \emph{succeeds}.
Success rates and failure rates are the percentage of instances that succeed or fail, respectively.

We evaluate the solution quality of a solved instance by computing a \emph{gap}
 as $\frac{d(T)}{d(T^\ast)} - 1$,
 the relative discrepancy between the cost of the found tree $T$ and the cost of the optimal Steiner tree $T^\ast$,
 usually given in thousandths (\permil).
When no optimal solution values are known,
 we use the currently best known upper bounds
 from the 11th DIMACS Challenge~\citeyearpar{DIMACS14Bounds}.

\begin{table}
 \tbl{%
  Mapping from \SteinLib instance groups to our instance grouping.
  \label{table:groups}
 }{%
  \scalebox{0.9}{%
  \begin{tabular}{|r|l|}
   \hline
            Group & \SteinLib group \\
   \hline
     EuclidSparse & \texttt{P6E} \\
   EuclidComplete & \texttt{P4E} \\
     RandomSparse & \texttt{B}, \texttt{C}, \texttt{D}, \texttt{E}; \texttt{P6Z}; non-complete instances from \texttt{MC} \\
   RandomComplete & \texttt{P4Z}; complete instances from \texttt{MC} \\
  IncidenceSparse & non-complete instances from \texttt{I080}, \texttt{I160}, \texttt{I320} and \texttt{I640} \\
IncidenceComplete & complete instances from \texttt{I080}, \texttt{I160}, \texttt{I320} and \texttt{I640} \\
ConstructedSparse & \texttt{PUC}; \texttt{SP} \\
SimpleRectilinear & \texttt{ES$\star$FST} and \texttt{TSPFST} with $|R| < 300$ or $|R|/|V| > 0.75$ \\
  HardRectilinear & \texttt{ES$\star$FST} and \texttt{TSPFST} with $|R| \geq 300$ and $|R|/|V| \leq 0.75$ \\
      VLSI / Grid & \texttt{ALUE}, \texttt{ALUT}, \texttt{LIN}, \texttt{TAQ}, \texttt{DIW}, \texttt{DMXA}, \texttt{GAP}, \texttt{MSM}; \texttt{1R}, \texttt{2R} \\
      WireRouting & \texttt{WRP3}, \texttt{WRP4} \\
   \hline
  \end{tabular}
 }}
\end{table}

Besides the original \SteinLib instance groups,
 we also consider a slightly different grouping
 where suitable, to obtain fewer but internally more consistent graph classes,
 see Table~\ref{table:groups} for details.
Additionally,
 \emph{Large} consists of all instances with more than 16\,000 edges or 8\,000 nodes,
 \emph{Difficult} are instances that could not be solved to proven optimality within one hour
  (according to the information provided by \SteinLib),
 and
 \emph{NonOpt} are 35 instances
  (31 from \texttt{PUC} and 4 from \texttt{I640})
  we still do not know the optimal values for.
Last but not least,
 we grouped instances by terminal coverage:
 the group \emph{Coverage $X$}
  contains all instances with $X - 10 < 100\, |R|/|V| \leq X$.

\subsection{Evaluation of $2$-Approximations}
\label{sec:exp-eval:2approx}

We consider the basic $2$-approximations by
 \citet{TM80}~(\algo{TM}, with \citep{AW02}),
 \citet{KMB81}~(\algo{KMB}),
 and \citet{M88}~(\algo{M}).
\algo{TM} and \algo{M} succeed for all instances,
 whereas \algo{KMB} fails for eight instances (mainly from \texttt{TSPFST})
 due to the memory limit.
Every instance is solved in less than 0.1 seconds using \algo{M},
 less than 0.3 seconds using \algo{TM},
 and at most 52 seconds using \algo{KMB} (instance \texttt{alue7080}).

Comparing the solution quality of
 \algo{TM}, \algo{KMB}, and \algo{M}
 gives further insights.
\algo{TM} yields significantly better solutons than \algo{M}:
 80.3\,\% of the solutions are better using \algo{TM}
 and only 1.7\,\% are worse.
Especially wire-routing instances
 are solved much better using \algo{TM}.
For example,
 the optimal solution for \texttt{wrp3-60}
 is 6\,001\,164,
 and it is solved to 6\,001\,175 (gap 0.0018\,\permil) using \algo{TM}
 and to 11\,600\,427 (gap 933.03\,\permil, \ie almost factor $2$) using \algo{M}.
A comparison between \algo{TM} and \algo{KMB} gives similar results.
On average, gaps are
 74.1\,\permil{} for \algo{TM},
 194.5\,\permil{} for \algo{M},
 and
 198.3\,\permil{} for \algo{KMB}.
\algo{TM} solves 9.2\,\% of the instances to optimality;
 \algo{M} only 3.9\,\% and \algo{KMB} 3.8\,\%.

We see that \algo{TM} is the best candidate:
 although slower than \algo{M}, it takes only negligible time,
 and solution quality is almost always better than for \algo{M} or \algo{KMB}.
When considering $2$-approximations in the following,
 we will always choose \algo{TM}.
In particular,
 we set \algo{TM} to be the final $2$-approximation to incorporate contracted full components (see Section~\ref{sec:GCF}).

\subsection{Evaluation of Full Component Enumeration}
\label{sec:exp-eval:comp-enum}

\paragraph{Shortest path algorithms.}
Before we consider the enumeration of full components, we compare the running times to compute
 the shortest path matrices using \texttt{sp=forbid} and \texttt{sp=prefer}.
In combination with \texttt{dist=sssp},
 these running times differ by at most 0.5 seconds for 97.8\,\% of the instances.
The only four outliers with more than 10 seconds time difference
 (maximum 120 seconds)
 are precisely the instances
 with a terminal coverage $|R|/|V| \geq 0.25$ and more than 8\,000 nodes.
In combination with \texttt{dist=apsp},
 there are already 20 instances where \texttt{forbid} can save between 15 seconds and 6 minutes,
 but the differences are still negligible for 87.1\,\% of the instances.
Hence, for the majority of the instances,
 \texttt{forbid} does not provide a significant time saving.

In contrast to this,
 we are more interested in the number of valid shortest paths that are obtained by the different variants.
A smaller (but sufficient) set of valid paths results in fewer full components.
On average,
 \texttt{forbid} generates a path between 76.3\,\% of all terminals pairs;
 \texttt{prefer} only between 59.1\,\%.
In particular, \texttt{forbid} could not reduce the number of valid paths at all
 for 65.8\,\% of the instances.
This number drops to 37.6\,\% using \texttt{prefer}, which is much better.
Consequently, \texttt{prefer} yields 10.5\,\% fewer $3$-restricted components than \texttt{sp=forbid}.
Since this saves memory and time for the further steps of the algorithm, we use \texttt{sp=prefer} in the following.

For $k = 3$,
 we can either use \texttt{dist=sssp} or \texttt{dist=apsp}.
While \texttt{sssp} is able to compute the shortest path matrices for every instance of the \SteinLib,
 all instances with more than 15\,000 nodes fail using \texttt{apsp}
 since the full APSP matrix is too big to fit into 16 GB of memory.

After filtering out all instances with negligible running times %
 using both algorithms,
 we obtain the following rule of thumb:
 use \texttt{sssp}
 iff
  the graph is not too dense (say, density $E/\binom{|V|}{2} \leq 0.25$). %
There are only a handful of outliers (in the I640 set) with small differences of up to $0.2$ seconds.%
\footnote{%
 Be aware that our rule suffices for the \SteinLib
  but is unlikely to hold as a general rule.
 While there are nearly all kinds of terminal coverages,
  the density distribution of the \SteinLib instances is quite unbalanced:
  there are no non-complete instances with density larger than $0.2$
  and most of the instances are sparse.}
We hence apply this rule to our experiments.

\paragraph{Full component generation.}
Fig.~\ref{fig:enum-performance}
 shows the percentage of instances
 such that the $3$-restricted components can be generated
 in a given time,
 for each \texttt{gen}- and \texttt{sp}-variant.
Note that \texttt{gen=all:smart} and \texttt{gen=all:naïve}
 use the same observation for $k=3$ and are hence equal.
It can clearly be observed
 that \texttt{gen=voronoi} with \texttt{sp=prefer} (as said above)
 is the best choice for $k=3$.

Fig.~\ref{fig:comp-enum}
 shows the success rates for the generation of $k$-restricted components
 for $k \in \set{3,4,5,6}$.
We can see that \texttt{gen=all:dw} is superior for $k \geq 4$.

\begin{figure}
\centering
 \includegraphics{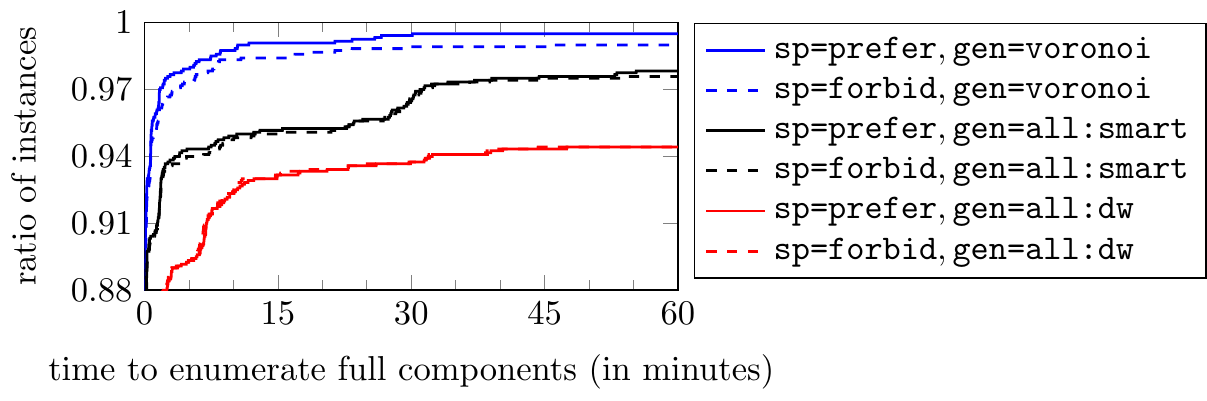}
\caption{Percentage of instances whose $3$-restricted components are generated within the given time
 for different variants.}
\label{fig:enum-performance}
\end{figure}

\subsection{Evaluation of GCF and LCA}
\label{sec:exp-eval:combinatorial}

We now consider the strong combinatorial algorithms with $k=3$.
Three instances failed for all algorithmic variants:
 \texttt{rl11849fst} (with 13\,963 nodes and 11\,849 terminals) is the only instance that failed due to the memory limit;
 \texttt{es10000fst01} and \texttt{fnl4461fst} failed due to the time limit.

\paragraph{Strategies for GCF.}
We first have a look at the reduction strategy \texttt{reduce=on}.
A comparison for $\win_\text{abs}$ with \texttt{gen=voronoi}
 shows that without reduction,
 we generate 68\,270 $3$-components on average,
 whereas
 we generate only 6\,136 $3$-components
 with reduction;
 in other words, 91.0\,\% of the generated full components
  are not promising
  and are hence removed directly after construction.
Considering the actual contractions,
 we effectively use 7.5\,\% of the generated unreduced full component set
 but 22.2\,\% of the reduced set.
The biggest time savings
 can be observed
 on rectilinear instances,
 \eg
 approximately half an hour for \texttt{fl3795fst}.
Since \texttt{reduce=on} offers benefits without introducing overhead,
 it will always be enabled in the following.

\begin{figure}
\centering
 \includegraphics{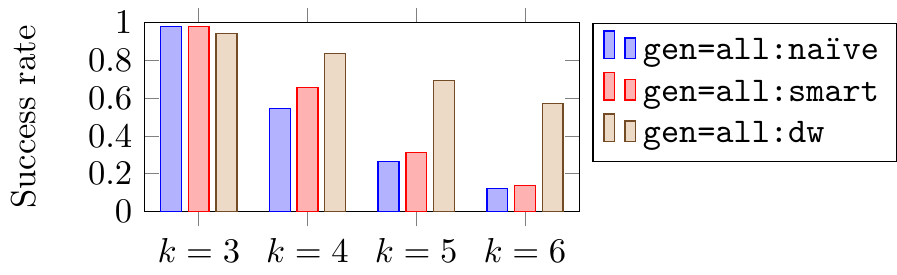}
 \caption{Success rates of component enumeration methods \texttt{gen=all}.}
 \label{fig:comp-enum}
\end{figure}

\begin{figure}
 \centering
 \subfigure[%
   Comparison between \texttt{gen=ondemand} and \texttt{voronoi},
    each with \texttt{save=static}.
  ]{\scalebox{0.75}{%
  \label{fig:z116-chart:ondemand-vs-voronoi}
  \includegraphics{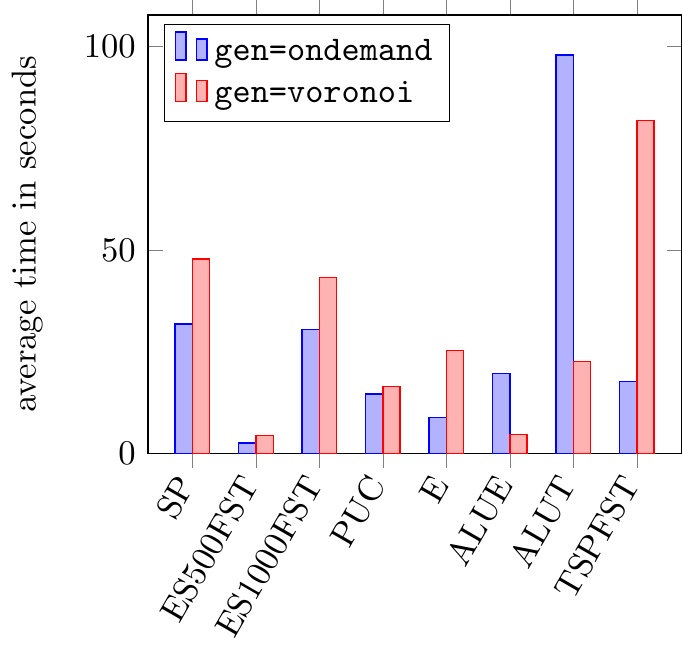}
 }}
 \qquad
 \subfigure[%
   Comparison between \texttt{save} strategies,
    each with $\win_\text{rel}$ and \texttt{gen=voronoi}.
  ]{\scalebox{0.75}{%
  \label{fig:z116-chart:save-comparison}
  \includegraphics{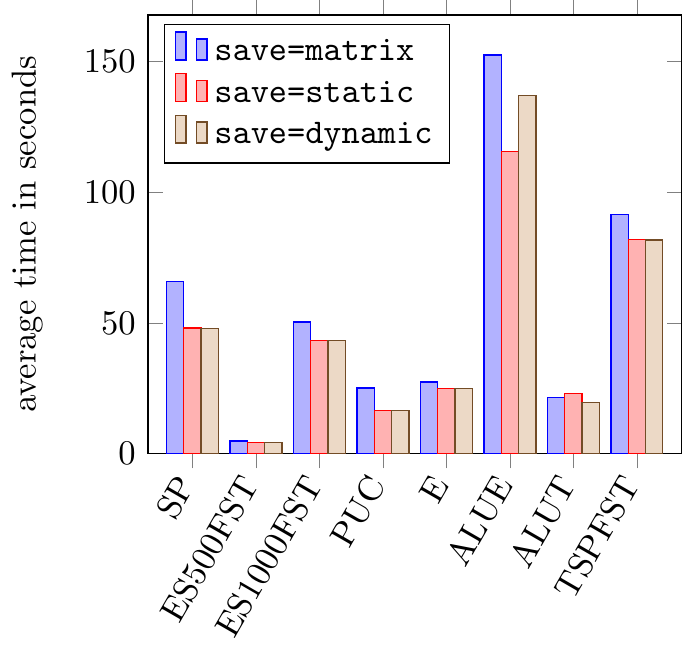}
 }}
 \caption{%
  Running times
   of GCF
   for different strategies,
   grouped by \SteinLib instance groups,
   averaged over the instances where all considered strategies succeed.
  Different choices of secondary strategies show similar tendencies.
  Omitted instance groups have negligible running times.
 }
 \label{fig:z116-chart}
\end{figure}

However, \texttt{gen=voronoi} is not the only way
 to generate full components for $\win_\text{abs}$ and $k=3$.
We can also use \texttt{gen=ondemand}.
This strategy
 is often but not always beneficial in comparison to \texttt{gen=voronoi}.
Let us, for example, consider GCF with $\win_\text{abs}$ and \texttt{save=static}.
(Numbers for other choices of \texttt{save} are similar.)
The average computation time decreases from 7.2 to 2.5 seconds.
However, extreme examples where one or the other choice is better, are
 \texttt{alut2625} with 192 seconds for \texttt{voronoi} and 849 seconds for \texttt{ondemand},
 and
 \texttt{fl3795fst} with 1\,782 and 148 seconds, respectively.
See Fig.~\ref{fig:z116-chart:ondemand-vs-voronoi} for a comparison by instance groups.
In general we might prefer \texttt{ondemand}
 but \texttt{voronoi} is the better choice for VLSI instances (like \texttt{ALUE}, \texttt{ALUT} and \texttt{LIN}).

We now compare the different strategies to compute \emph{save}.
Averaged over all instances,
 the times for \texttt{save=static} and \texttt{save=dynamic}
 are nearly the same.
Strategy \texttt{save=static} is never (significantly) worse than \texttt{save=matrix}:
 the average time decreases from 10.3 (\texttt{matrix}) to 8.6 seconds (\texttt{static}) for \texttt{gen=voronoi}
 (and from 6.6 to 4.2 seconds for \texttt{gen=ondemand}).
See Fig.~\ref{fig:z116-chart:save-comparison} for a more detailed view based on \SteinLib instance groups.
Since \texttt{static} is easier to implement than \texttt{dynamic},
 we recommend to use \texttt{save=static}.

Surprisingly,
 the strategy \texttt{singlepass=on}
 does not lead to notable time improvements
 in comparison to the original algorithm.
The solution becomes worse in
 33.3\,\% (28.4\,\%)
 and better in 22.7\,\% (18.8\,\%)
 of the instances
 using $\win_\text{abs}$ ($\win_\text{rel}$, respectively),
 deteriorating the solutions by 0.13\,\permil{} (0.28\,\permil) on average.
Hence,
 we will not consider \texttt{singlepass=on} in the following.

\begin{table}
 \tbl{%
  Comparison of combinatorial algorithms by instance groups.
  Per group, we give the
   total number of instances (`\#'),
   portion of optimally solved instances,
   average gaps,
   and average solution times.
  \label{table:gcf3-comparison}
 }{%
 \scalebox{0.6}{%
 \setlength\tabcolsep{2pt}
\begin{tabular}{|r|r||rrrrr|rrrrr||rrrrr|}
 \hline
  &  & \multicolumn{5}{c|}{Optimals\%} & \multicolumn{5}{c||}{Average gap\permil} & \multicolumn{5}{c|}{Average time in \textit{sec}} \\
 Group & \# & \algo{TM} & \algo{ACO} & \algo{AC3} & \algo{RC3} & \algo{LC3} & \algo{TM} & \algo{ACO} & \algo{AC3} & \algo{RC3} & \algo{LC3} & \algo{TM} & \algo{ACO} & \algo{AC3} & \algo{RC3} & \algo{LC3} \\
 \hline
     EuclidSparse &        15 &       33.3 & \bf   93.3 &       86.7 &       80.0 &       66.7 &        15.18 & \bf     1.92 &         1.99 &         2.00 &         5.27 &         0.00 & \bf     0.00 & \bf     0.00 & \bf     0.00 & \bf     0.00 \\
   EuclidComplete &        14 &        7.1 & \bf   78.6 & \bf   78.6 & \bf   78.6 & \bf   78.6 &        10.84 & \bf     0.35 & \bf     0.35 & \bf     0.35 & \bf     0.35 &         0.01 & \bf     0.08 & \bf     0.08 & \bf     0.08 &         3.98 \\
     RandomSparse &        96 &       27.1 &       41.7 &       40.6 & \bf   43.8 &       40.6 &        23.72 &        12.32 &        12.33 & \bf    10.93 &        18.10 &         0.01 & \bf     1.96 &         5.55 &         5.50 &        24.63 \\
   RandomComplete &        13 &       30.8 & \bf   61.5 & \bf   61.5 & \bf   61.5 &       53.8 &        22.34 & \bf    20.55 &        37.56 &        39.73 &        45.03 &         0.00 & \bf     0.01 & \bf     0.01 & \bf     0.01 &         0.06 \\
  IncidenceSparse &       320 &        3.8 &        4.4 &        4.4 & \bf    5.1 &        3.5 &       124.13 &        70.66 &        70.04 & \bf    68.94 &        75.16 &         0.00 & \bf     0.05 & \bf     0.03 & \bf     0.03 &         1.91 \\
IncidenceComplete &        80 &        0.0 &        0.0 &        0.0 &        0.0 &        0.0 &       384.29 & \bf   126.37 &       127.30 &       128.71 &       126.65 &         0.06 & \bf     0.26 & \bf     0.28 & \bf     0.29 &        40.90 \\
ConstructedSparse &        58 & \bf   25.9 & \bf   25.9 &       22.2 &       22.2 & \bf   25.9 & \bf    95.52 &       129.91 &       148.82 &       150.74 &       137.96 &         0.00 & \bf    17.02 &        20.85 &        20.96 &        73.93 \\
SimpleRectilinear &       218 &       11.5 &       17.4 & \bf   18.3 &       17.9 &       15.6 &        16.96 &         6.63 &         6.40 & \bf     6.23 &         9.91 &         0.00 & \bf     0.87 &        16.00 &        16.03 &        17.76 \\
  HardRectilinear &        54 &        0.0 &        0.0 &        0.0 &        0.0 &        0.0 &        22.86 &         8.58 &         8.47 & \bf     8.05 &        13.37 &         0.00 & \bf    30.27 &        59.82 &        59.83 &        58.10 \\
      VLSI / Grid &       207 &       10.6 & \bf   35.7 &       35.3 &       33.8 &       29.0 &        33.83 &         9.75 &         9.94 & \bf     9.74 &        24.17 &         0.00 &         5.99 & \bf     1.53 & \bf     1.53 &         2.12 \\
      WireRouting &       125 &        3.2 & \bf    4.0 &        2.4 &        3.2 &        3.2 & \bf     0.01 & \bf     0.01 & \bf     0.01 & \bf     0.01 &         0.02 &         0.00 &         0.14 & \bf     0.07 & \bf     0.07 &         0.60 \\
\hline
            Large &       187 &        1.8 &        5.4 &        5.4 & \bf    6.0 & \bf    6.0 &       230.86 & \bf    95.58 &        98.43 &        98.29 &       100.04 &         0.04 &        13.79 & \bf    10.72 & \bf    10.72 &        58.83 \\
        Difficult &       146 &        2.7 &        2.7 &        1.8 &        2.7 & \bf    3.6 &       206.46 & \bf   110.13 &       116.83 &       117.39 &       113.30 &         0.03 & \bf     7.55 &         8.89 &         8.94 &        60.16 \\
           NonOpt &        35 &        --- &        --- &        --- &        --- &        --- & \bf   134.88 &       153.86 &       169.94 &       170.53 &       156.79 &         0.01 & \bf    20.95 &        23.61 &        23.74 &       112.92 \\
\hline
      Coverage 10 &       588 &        9.2 & \bf   19.9 &       19.6 & \bf   19.9 &       16.4 &        73.05 &        38.69 &        39.11 & \bf    38.05 &        46.97 &         0.01 &         2.18 & \bf     0.59 & \bf     0.59 &         1.18 \\
      Coverage 20 &       159 &        9.6 & \bf   18.5 &       16.4 &       15.8 &       14.4 &       111.47 & \bf    57.03 &        57.89 &        59.16 &        61.03 &         0.00 &         1.13 & \bf     0.40 & \bf     0.40 &        22.89 \\
      Coverage 30 &       127 &        2.4 &        4.1 &        4.1 &        4.1 & \bf    5.7 &       186.25 & \bf    67.17 &        67.62 &        68.17 &        69.12 &         0.02 &         2.22 & \bf     1.94 & \bf     1.93 &        37.80 \\
      Coverage 40 &        92 &        1.1 & \bf    3.3 & \bf    3.3 & \bf    3.3 &        2.2 &        26.19 &         9.55 &         9.27 & \bf     8.82 &        15.60 &         0.00 & \bf     5.94 &         8.51 &         8.52 &         8.19 \\
      Coverage 50 &       135 &        5.5 &       18.1 &       18.1 & \bf   18.9 &       16.5 & \bf    23.80 &        27.40 &        32.54 &        32.82 &        30.03 &         0.00 & \bf     5.04 &         9.23 &         9.23 &        13.36 \\
      Coverage 60 &        45 &       11.1 &       26.7 & \bf   31.1 &       24.4 &       24.4 &        18.88 & \bf    16.27 &        17.64 &        18.33 &        18.31 &         0.00 & \bf    20.49 &        36.58 &        36.66 &        36.08 \\
      Coverage 70 &        18 &       44.4 & \bf   50.0 &       44.4 &       44.4 &       38.9 &         6.37 & \bf     1.98 &         2.44 &         2.14 &         3.76 &         0.00 & \bf     4.91 &        24.19 &        24.07 &        23.32 \\
      Coverage 80 &        11 & \bf   54.5 & \bf   54.5 & \bf   54.5 & \bf   54.5 & \bf   54.5 &         5.12 &         2.36 & \bf     2.32 &         2.43 &         2.45 &         0.00 & \bf    21.36 &       222.85 &       223.27 &       259.47 \\
      Coverage 90 &        14 &       21.4 &       28.6 &       28.6 &       35.7 & \bf   42.9 &         3.87 &         0.70 &         0.78 &         0.69 & \bf     0.68 &         0.00 & \bf     0.99 &        23.76 &        23.69 &        23.09 \\
     Coverage 100 &        11 &       54.5 &       63.6 &       63.6 & \bf   72.7 &       63.6 &         0.94 & \bf     0.07 &         0.25 &         0.27 &         0.23 &         0.00 & \bf     0.93 &       116.49 &       116.70 &       112.91 \\
\hline
              All &      1200 &        9.1 & \bf   18.1 &       17.8 &       17.9 &       15.7 &        76.06 & \bf    38.23 &        39.22 &        38.95 &        43.91 &         0.01 & \bf     3.54 &         7.24 &         7.24 &        15.02 \\
\hline
\end{tabular}
 }}
\end{table}

\paragraph{Comparison.}
We compare
 the following algorithms for $k=3$:
 GCF with $\win_\texttt{abs}$ and \texttt{gen=voronoi} (\algo{AC3}),
 GCF with $\win_\texttt{abs}$ and \texttt{gen=ondemand} (\algo{ACO}),
 GCF with $\win_\texttt{rel}$ (\algo{RC3}),
 and
 LCA with $\win_\texttt{loss}$ (\algo{LC3}).
Later we will also consider \algo{AC$k$}, \algo{RC$k$}, and \algo{LC$k$}
 for $k \geq 4$
 where \texttt{gen=all:dw} will be used.
Success rates are
 99.7\,\% for \algo{ACO}
 and 99.5\,\% for \algo{AC3}, \algo{RC3}, and \algo{LC3}.
We observe that these success rates are clearly dominated by the full component generation method.
See Table~\ref{table:gcf3-comparison} for a comparison
 of solution quality and time consumption
 based on our instance groups.
With respect to solution quality,
 all algorithms are worthwile:
 in comparison to \algo{TM},
 the average gap halves
 and the number of optimally solved instances doubles.
Noteworthy exceptions are
 \emph{RandomComplete}
  where only \algo{ACO} provides (slightly) better average gaps than \algo{TM},
 and
 \emph{ConstructedSparse},
  where \algo{TM} provides the best average gaps.

Among the strong algorithms \algo{ACO}, \algo{AC3}, \algo{RC3}, and \algo{LC3},
 there is no clear winner
 regarding solution quality
 for a majority of the instances.
\algo{RC3} is better than \algo{AC3} on average
 within almost identical running time.
\algo{ACO} is almost always better than \algo{AC3};
 if not, it is only slightly worse.
We emphasize that
 \algo{ACO} becomes---in comparison---significantly faster
 for increasing terminal coverage,
 and outperforms the other algorithms already for $|R|/|V| > 0.2$.
Although \algo{LC3} takes significantly more time
 than the other algorithms,
 the obtained solution quality is not significantly better.
\algo{ACO} offers the best compromise between time and solution quality.

\subsection{Evaluation of the LP-based Algorithm}
\label{sec:exp-eval:lpbased}

The main stages of the LP-based algorithm are
 (1)~the full component enumeration,
 (2)~solving the LP relaxation,
 and
 (3)~the approximation based on the fractional LP solution.
We have already evaluated the strategies for~(1) in Section~\ref{sec:exp-eval:comp-enum},
 and will now evaluate the different strategies for the remaining stages,
 for $k = 3$.

\paragraph{Solving the LP relaxation.}
Fig.~\ref{fig:separation-performance}
 shows that
 \texttt{consep=on}
 is clearly beneficial.
Together with
 \texttt{presep=ondemand},
 it allows us to compute the LP solutions
 for 89.0\,\% of the instances;
 a slight improvement over the 88.9\,\% with \texttt{presep=initial}.
We hence perform all further experiments using
 connectivity tests and by separating constraints~\eqref{eq:ser:yv}.

\begin{figure}
\centering
 \includegraphics{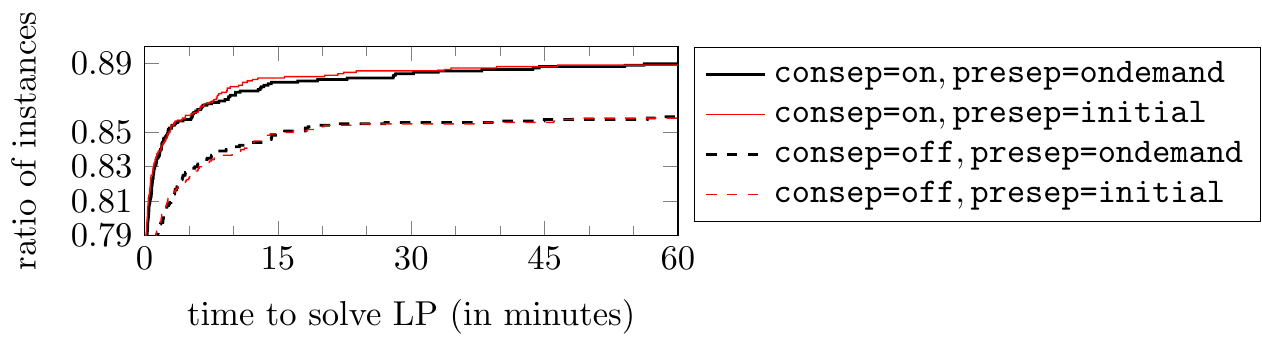}
\caption{Percentage of instances whose LP relaxations for $k=3$ are solved within the given time
 for different variants.}
\label{fig:separation-performance}
\end{figure}

\paragraph{Strategies.}
Using the original algorithm,
 all instances that pass stage (2) also pass stage (3) of the algorithm.
That leads to the assumption that stages (1) and (2) are the dominating stages.

The distribution of time in the three stages
 is very different for different instances.
We consider the 133 instances with more than 10~seconds computation time.
On average, we spend 8.7\,\% of the time in~(1),
 90.4\,\% in~(2)
 and 0.8\,\% in~(3).
Stage~(1) dominates in 7.5\,\% of the instances.
An extreme example is \texttt{u2152fst}
 where the algorithm spends 280 seconds in~(1),
  40~seconds in~(2),
  and 85~seconds in~(3).
In the remaining 92.5\,\% of the instances, stage~(2) dominates.
There are many instances, especially in the \texttt{*FST} set,
 with a long running time in~(2)
 but negliglible running times for~(1) and~(3);
 the most extreme example is \texttt{linhp318fst}
 spending 57 minutes in~(2).
Stage~(3) never dominates.
The most extreme example is~\texttt{d2103fst},
 where it spends 4 minutes in~(1),
 36 minutes in~(2)
 and almost 2 minutes in~(3).
This means that%
 ---as for the previous combinatorial algorithms---%
 the \emph{actual} approximation algorithm
 needs the least of the whole running time.

The strategy \texttt{prune=on}
 achieves that the time of stage (3) becomes negligible for \emph{every} instance.
Although the overall effect of that strategy may be considered rather limited
 (since usually most of the time is not spent in the last stage anyhow),
 we apply this strategy in the following.

\begin{table}
 \tbl{%
  Comparison of algorithms by instance groups.
  Per group, we give the
   total number of instances (`\#'),
   success rates,
   portion of optimally solved instances,
   average gaps,
   portion of instances where \algo{LP3} obtained a worse or equal solution than \algo{TM} or \algo{ACO},
   and average solution times.
  \label{table:lp3-comparison}
 }{%
 \scalebox{0.6}{%
 \setlength\tabcolsep{2pt}
\begin{tabular}{|r|r||rrr|rrr|rrr|rr||rrr|}
 \hline
  &  & \multicolumn{3}{c|}{Success\%} & \multicolumn{3}{c|}{Optimals\%} & \multicolumn{3}{c|}{Average gap\permil} & \multicolumn{2}{c||}{\algo{LP3}$\geq$} & \multicolumn{3}{c|}{Avg. time in \textit{sec}} \\
 Group & \# & \algo{ACO} & \algo{LP3} & \algo{BC} & \algo{TM} & \algo{ACO} & \algo{LP3} & \algo{TM} & \algo{ACO} & \algo{LP3} & \algo{TM} & \algo{ACO} & \algo{ACO} & \algo{LP3} & \algo{BC} \\
 \hline
     EuclidSparse &       15 &     100.0 &     100.0 &     100.0 &      33.3 & \bf  93.3 &      80.0 &       15.18 & \bf    1.92 &        3.31 &      33.3 & \bf  93.3 &        0.00 & \bf    0.02 &        0.29 \\
   EuclidComplete &       14 &     100.0 & \bf 100.0 &      92.9 &       7.1 & \bf  78.6 & \bf  78.6 &       10.84 & \bf    0.35 & \bf    0.35 &       7.1 & \bf 100.0 &        0.01 & \bf    0.12 &       21.29 \\
     RandomSparse &       96 &     100.0 &      75.0 & \bf 100.0 &      27.1 & \bf  41.7 &      33.3 &       23.62 & \bf   13.16 &       17.01 &      66.7 & \bf  92.7 &        0.03 &       88.36 & \bf    0.92 \\
   RandomComplete &       13 &     100.0 &     100.0 &     100.0 &      30.8 & \bf  61.5 &      46.2 &       22.34 & \bf   20.55 &       37.14 & \bf  76.9 & \bf 100.0 &        0.01 & \bf    0.36 &        6.78 \\
  IncidenceSparse &      320 &     100.0 & \bf  93.8 &      88.1 &       3.8 & \bf   4.4 &       3.2 &      118.71 & \bf   70.04 &       75.35 &      44.1 &      69.7 &        0.01 & \bf    2.10 &      120.22 \\
IncidenceComplete &       80 &     100.0 & \bf  93.8 &      81.2 &       0.0 &       0.0 &       0.0 &      378.80 & \bf  125.72 &      140.93 &       6.2 & \bf  85.0 &        0.14 & \bf    0.61 &      129.74 \\
ConstructedSparse &       58 &     100.0 & \bf  67.2 &      25.9 & \bf  25.9 & \bf  25.9 &      18.5 & \bf   84.88 &      108.32 &      132.50 & \bf  87.9 & \bf  86.2 &        0.01 & \bf    7.05 &      147.98 \\
SimpleRectilinear &      218 &      99.5 &      95.9 & \bf  99.5 &      11.5 &      17.4 & \bf  21.1 &       17.16 &        6.69 & \bf    5.03 &      19.3 &      45.0 &        0.07 &       77.93 & \bf    0.42 \\
  HardRectilinear &       54 &      96.3 &       3.7 & \bf  94.4 &       0.0 &       0.0 &       0.0 &       21.99 &        8.57 & \bf    2.57 & \bf  96.3 & \bf  96.3 &        0.30 &     3321.97 & \bf    5.88 \\
      VLSI / Grid &      207 &      99.5 & \bf  98.6 &      86.5 &      10.6 & \bf  35.7 & \bf  35.7 &       33.82 & \bf    9.73 &       11.04 &      20.3 & \bf  83.6 &        0.05 & \bf    2.28 &      159.90 \\
      WireRouting &      125 &     100.0 & \bf 100.0 &      89.6 &       3.2 & \bf   4.0 &       3.2 & \bf    0.01 & \bf    0.01 & \bf    0.01 &      51.2 &      57.6 &        0.08 & \bf   13.85 &      182.54 \\
\hline
            Large &      187 &      97.9 & \bf  78.6 &      51.3 &       1.8 & \bf   5.4 &       4.8 &      241.62 & \bf   93.49 &      101.69 &      31.0 & \bf  85.6 &        0.18 & \bf    2.34 &      292.48 \\
        Difficult &      146 &      98.6 & \bf  74.0 &      25.3 & \bf   2.7 & \bf   2.7 &       1.8 &      198.75 & \bf   98.15 &      106.51 &      56.8 & \bf  80.8 &        0.18 & \bf   20.49 &     1192.83 \\
           NonOpt &       35 &     100.0 & \bf  42.9 &       0.0 &       --- &       --- &       --- & \bf  115.58 &      141.40 &      159.32 & \bf  88.6 & \bf  85.7 &         --- &         --- &         --- \\
\hline
      Coverage 10 &      588 &      99.8 & \bf  99.3 &      87.9 &       9.2 & \bf  19.9 &      19.2 &       73.03 & \bf   38.37 &       43.31 &      37.2 &      74.8 &        0.05 & \bf    3.30 &      124.23 \\
      Coverage 20 &      159 &     100.0 & \bf  91.2 &      84.9 &       9.6 & \bf  18.5 &      11.6 &      112.95 & \bf   55.65 &       62.30 &      47.2 & \bf  78.6 &        0.03 & \bf   12.07 &      104.40 \\
      Coverage 30 &      127 &      99.2 &      71.7 & \bf  76.4 &       2.4 & \bf   4.1 &       3.3 &      183.84 & \bf   66.71 &       69.44 &      39.4 & \bf  79.5 &        0.02 & \bf   13.62 &      193.39 \\
      Coverage 40 &       92 &      98.9 &      65.2 & \bf  98.9 &       1.1 &       3.3 & \bf   4.3 &       26.70 &        9.76 & \bf    7.78 &      35.9 &      55.4 &        0.04 &       81.67 & \bf    0.86 \\
      Coverage 50 &      135 &     100.0 &      83.7 & \bf  91.1 &       5.5 & \bf  18.1 &      15.7 &       20.92 & \bf   17.11 &       17.98 &      32.6 &      54.8 &        0.03 &      138.78 & \bf    0.88 \\
      Coverage 60 &       45 &     100.0 &      82.2 & \bf  93.3 &      11.1 &      26.7 & \bf  31.1 &       11.57 & \bf    6.17 &        6.42 &      37.8 &      66.7 &        0.00 &       52.41 & \bf    1.94 \\
      Coverage 70 &       18 &     100.0 &      66.7 & \bf 100.0 &      44.4 & \bf  50.0 &      44.4 &        2.27 & \bf    0.68 &        2.20 & \bf  77.8 & \bf  88.9 &        0.00 &        0.42 & \bf    0.07 \\
      Coverage 80 &       11 &     100.0 &      72.7 & \bf  90.9 & \bf  54.5 & \bf  54.5 & \bf  54.5 &        2.22 &        0.55 & \bf    0.38 & \bf  81.8 & \bf  81.8 &        0.03 &       50.70 & \bf    0.07 \\
      Coverage 90 &       14 &      92.9 &      57.1 & \bf 100.0 &      21.4 &      28.6 & \bf  42.9 &        3.60 &        0.53 & \bf    0.09 &      64.3 & \bf  85.7 &        0.06 &       69.26 & \bf    0.10 \\
     Coverage 100 &       11 &     100.0 &      90.9 & \bf 100.0 &      54.5 &      63.6 & \bf  90.9 &        0.94 &        0.07 & \bf    0.00 &      63.6 &      72.7 &        0.93 &      486.39 & \bf    0.17 \\
\hline
              All &     1200 &      99.7 & \bf  89.0 &      88.2 &       9.1 & \bf  18.1 &      17.2 &       75.13 & \bf   36.81 &       40.65 &      39.8 &      72.2 &        0.05 & \bf   33.09 &       96.07 \\
\hline
\end{tabular}
 }}
\end{table}

The strategy \texttt{stronger=on} might not be as worthwhile:
 the success rate drops to 88.8\,\%.
However, 68.0\,\% of the solved LP relaxations are solved integrally with \texttt{stronger=on}
 whereas it is only 50.0\,\% for the original LP relaxation.
Although this sounds promising,
 the final solution of the majority (74.2\,\%) of the instances does not change
 and the solution value increases (becomes worse) for 17.6\,\% of the instances.
A harsh example is \texttt{i160-043}
 where the fractional LP solution increases by $1.5$
 but the integral approximation increases from 1\,549 to 1\,724.
Only 8.2\,\% of the solutions improve.
A good example is \texttt{i080-003}
 where the fractional LP solution 1\,902 increases to the (integral) LP solution 1\,903
 and yields an improvement from 1\,807 to the optimum Steiner tree value 1\,713.
We mention these examples
 to conclude that the observed integrality gap
 of the relaxation's solution
 seems secondary.
The primary influence for the solution quality seems to be the actual choice of full components in the fractional solution,
 and the choice of core edges.
We hence refrain from using \texttt{stronger=on}
 in the following and cannot recommend it.

We now evaluate \texttt{bound=on}.
The success rate increases to 92.2\,\%.
LP solving times improve for 43.1\,\% of the instances
 and deteriorate for 10.8\,\%,
 solutions improve for 49.5\,\%
 and deteriorate for 19.3\,\%.
This behavior is due to the fact that
 the LP solver hits the bound on 67.8\,\% of the successful instances
 and then just returns the $2$-approximation.
---
 Although \texttt{bound=on} is worthwhile in practice, we
will not use it in the following comparison.
It can be seen as an aggregation to 
combine two distinct methods, whereas in the following, we want to 
see a clearer picture on the differences between the various methods.

\paragraph{Comparison.}
We compare the LP-based approximation algorithm~(\algo{LP$k$})
 with $k=3$
 to the recommended $2$-approximation \algo{TM},
 the $11/6$-approximation \algo{ACO},
 and an exact algorithm:
 \algo{BC}, a highly-tuned branch-and-cut approach
 presented by \citet{FLLLMRSS15}.
\algo{BC} has been one of the winners of the \citet{DIMACS14Bounds}.
It is based on an integer linear program
 that%
 ---using a branch-and-cut framework---%
 is arguably much easier to implement than,
 \eg
 the sophisticated strong approximation algorithms.

See Table~\ref{table:lp3-comparison} for a comparison based on our instance groups.
In comparison to \algo{TM} only,
 \algo{LP3} achieves a significantly better solution quality
  for most of the instances it can solve,
  and the average running times may be justifiable.
However,
 the much simpler algorithm \algo{ACO}
 is clearly better in terms of time \emph{and} solution quality:
 its solutions are not worse than \algo{LP3} solutions
  in 72.2\,\% of the instances.
\algo{BC} fails for insignificantly more instances than \algo{LP3}.
The results suggest
 to use \algo{BC} for rectilinear instances
 and instances with terminal coverage $|R|/|V| \geq 0.4$.

\subsection{Higher $k$}

Finally,
 we consider the strong approximation algorithms for higher $k$.
  see Table~\ref{table:higher-k-comparison}.
The success rates for $k = 6$ drop below 60\,\%
 which can be considered clearly impractical.
Hence we compare the key performance indicators for
 the instances where \algo{AC$k$}, \algo{RC$k$}, \algo{LC$k$}, and \algo{LP$k$}
 succeed for $k = 3,4,5$.
(Recall that \algo{ACO} does not generalize to $k > 3$.)
All algorithms improve their solution qualities with increasing $k$.
This is surprising for \algo{AC$k$}
 since it is only proven to be a $11/6$-approximation for any $k \geq 3$.
On the other hand,
 $11/6$ is smaller than the theoretically proven bounds of the other algorithms for $k \leq 5$.
\algo{AC$k$}, \algo{RC$k$}, and \algo{LP$k$}
 are comparable regarding their average gaps;
 only \algo{LC$k$} turns out to be worse than the others.
\algo{RC$k$} always achieves the smallest gaps
 and also has the smallest or reasonable small running times.
This is remarkable since,
 according to theoretical bounds,
 \algo{RC$k$} is the worst choice for $k \leq 18$.
\algo{LP$k$} has the highest chances to find the optimum.
However, if it fails to find the optimum,
 the solution is either quite weak (worse gaps than \algo{AC$k$} and \algo{RC$k$})
 or non-existent (worst success rates).

In Table~\ref{table:k4-comparison},
 we give a more in-depth look at the solution qualities of the algorithms with $k=3$ and $k=4$ only,
 since $k=4$ still offers good success rates.
Interestingly,
 for \algo{AC$k$}, \algo{RC$k$}, and \algo{LC$k$},
 the \emph{WireRouting} and \emph{NonOpt} solutions become slightly worse
 when we increase $k$ from $3$ to $4$
We can also see that the algorithms perform bad on high-coverage instances
 regarding success rates.
However, if successful, the solution quality is good,
 \eg
 \algo{LP4} is already a good choice for terminal coverage $|R|/|V| > 0.3$.

\begin{table}
 \tbl{%
  Comparison of strong approximation algorithms for higher $k$.
  Per $k$, we give the success rates for each algorithm,
   and then%
   ---limited to the 787 instances that could be solved by all algorithms with $k \in \set{3,4,5}$---%
   the portion of optimally solved instances,
   the average gaps,
   and the average running time.
 \label{table:higher-k-comparison}
 }{%
 \scalebox{0.7}{%
 \setlength\tabcolsep{2pt}
\begin{tabular}{|r||rrrr|rrrr|rrrr||rrrr|}
 \hline
       & \multicolumn{4}{c|}{Success\%}
       & \multicolumn{4}{c|}{Optimals\%}
       & \multicolumn{4}{c||}{Average gap\permil}
       & \multicolumn{4}{c|}{Average time in \textit{sec}}
       \\
 Alg   & \algo{AC$k$} & \algo{RC$k$} & \algo{LC$k$} & \algo{LP$k$}
       & \algo{AC$k$} & \algo{RC$k$} & \algo{LC$k$} & \algo{LP$k$}
       & \algo{AC$k$} & \algo{RC$k$} & \algo{LC$k$} & \algo{LP$k$}
       & \algo{AC$k$} & \algo{RC$k$} & \algo{LC$k$} & \algo{LP$k$}
       \\
 \hline
 $k=3$ & \bf     99.5 & \bf     99.5 & \bf     99.5 &         89.0
       & \bf     23.8 &         23.5 &         20.7 &         22.6
       &        38.99 & \bf    38.57 &        44.92 &        43.08
       & \bf     0.02 & \bf     0.02 &         0.08 &         0.16
 \\
 $k=4$ & \bf     85.2 & \bf     85.2 & \bf     85.2 &         77.2
       &         29.1 &         29.9 &         24.0 & \bf     38.2
       &        26.63 & \bf    25.89 &        34.16 &        26.90
       & \bf     9.77 &         9.97 & \bf     9.73 &        10.60
 \\
 $k=5$ & \bf     70.8 & \bf     70.8 &         70.6 &         65.7
       &         30.6 &         32.1 &         24.9 & \bf     46.3
       &        20.41 & \bf    19.42 &        30.22 &        20.75
       &       131.07 & \bf   128.24 &       129.08 &       130.89
 \\
 $k=6$ & \bf     57.4 & \bf     57.4 &         57.2 &         54.3
 & & & &
 & & & &
 & & & &
 \\
 \hline
\end{tabular}
 }}
\end{table}

\section{Conclusion}

We considered the strong approximation algorithms
 for the Steiner tree problem (STP)
 with an approximation ratio below~2.
While there has been many theoretical advances over the
 last decades \wrt\ the approximation ratio,
 their practical applicability and strength has never been considered.
In particular,
 all these algorithms use the tool of \emph{$k$-restricted components}
 as a central ingredient to achieve astonishing approximation ratios for $k\to\infty$,
 while the runtime is exponentially dependent on~$k$.
The concept hence turned out to be a main stumbling block in real-world applications
 since they are both time- and memory-consuming.
This paper is an attempt to show the importance of the research field
 \emph{algorithm engineering}.
Amongst other findings,
 we pinpoint further worthwhile research questions
 both from the theoretical and the practical point of view,
 and hope to increase the awareness for the necessity to complement high-level theoretical research
 with practical considerations,
 in order to ensure a certain degree of `groundedness' of the theory.

\begin{table}
 \tbl{%
  Comparison of solution qualities obtained by algorithms for $k=3$ and $k=4$.
  Per instance group, we give the
   number of instances (`\#')
    that could be solved by all considered algorithms,
   portion of optimally solved instances,
   and average gaps.
 \label{table:k4-comparison}
 }{%
 \scalebox{0.6}{%
 \setlength\tabcolsep{2pt}
\begin{tabular}{|r|rr||rr|rr|rr|rr||rr|rr|rr|rr|}
 \hline
       &    &    & \multicolumn{8}{c||}{Optimals\%} & \multicolumn{8}{c|}{Average gap\permil} \\
 Group & \# & \% & \algo{AC3} & \algo{AC4} & \algo{RC3} & \algo{RC4} & \algo{LC3} & \algo{LC4} & \algo{LP3} & \algo{LP4} & \algo{AC3} & \algo{AC4} & \algo{RC3} & \algo{RC4} & \algo{LC3} & \algo{LC4} & \algo{LP3} & \algo{LP4} \\
 \hline
     EuclidSparse  &   15 & 100   &     86.7 &     86.7 &     80.0 &     80.0 &     66.7 &     73.3 &     80.0 &\bf 100   &       1.99 &       1.59 &       2.00 &       1.40 &       5.27 &       4.95 &       3.31 & \bf   0.00 \\
     EuclidCompl.  &   13 &  92.9 &     84.6 &\bf 100   &     84.6 &\bf 100   &     84.6 &\bf 100   &     84.6 &\bf 100   &       0.36 & \bf   0.00 &       0.36 & \bf   0.00 &       0.36 & \bf   0.00 &       0.36 & \bf   0.00 \\
      Rand.Sparse  &   61 &  63.5 &     57.4 &\bf  72.1 &     60.7 &     70.5 &     57.4 &     57.4 &     50.8 &\bf  72.1 &      13.26 &       8.63 &      11.08 &       7.69 &      21.58 &      16.78 &      17.74 & \bf   7.17 \\
      Rand.Compl.  &   12 &  92.3 &\bf  66.7 &     58.3 &\bf  66.7 &     50.0 &     58.3 &\bf  66.7 &     50.0 &\bf  66.7 &      26.20 &      22.66 &      28.55 & \bf  19.70 &      35.19 &      27.24 &      26.64 &      23.12 \\
     Incid.Sparse  &  284 &  88.8 &      4.9 &      6.7 &      5.6 &      8.1 &      3.9 &      6.0 &      3.5 &\bf   9.2 &      68.51 &      50.18 &      67.29 & \bf  49.37 &      74.24 &      58.49 &      74.67 &      50.06 \\
     Incid.Compl.  &   71 &  88.8 &      0.0 &      0.0 &      0.0 &      0.0 &      0.0 &      0.0 &      0.0 &      0.0 &     125.76 &      78.01 &     127.43 & \bf  77.56 &     125.53 & \bf  77.52 &     140.58 &      95.26 \\
    Constr.Sparse  &   25 &  43.1 &     28.6 &\bf  33.3 &     28.6 &\bf  33.3 &\bf  33.3 &     23.8 &     23.8 &     28.6 &     100.90 & \bf  62.78 &     105.25 &      66.40 &      96.28 &      78.39 &     107.71 &      68.95 \\
      SimpleRect.  &  175 &  80.3 &     18.9 &     21.1 &     17.1 &     22.3 &     14.9 &     16.6 &     20.0 &\bf  47.4 &       6.69 &       5.08 &       6.55 &       3.82 &      10.35 &       8.14 &       5.21 & \bf   1.27 \\
      VLSI / Grid  &  180 &  87.0 &     40.6 &     52.2 &     38.9 &     53.3 &     33.3 &     42.2 &     41.1 &\bf  67.2 &       9.36 &       5.17 &       9.06 &       4.39 &      23.22 &      20.29 &      10.56 & \bf   2.49 \\
      WireRouting  &   91 &  72.8 &      3.3 &      5.5 &      4.4 &\bf   6.6 &      4.4 &      3.3 &      4.4 &\bf   6.6 &       0.01 &       0.02 &       0.01 &       0.02 &       0.02 &       0.02 &       0.01 &       0.01 \\
\hline
            Large  &  112 &  11.4 &      6.5 &      9.3 &      7.4 &      9.3 &      7.4 &      6.5 &      7.4 &\bf  12.0 &     103.13 &      63.07 &     102.49 & \bf  61.82 &     105.14 &      67.84 &     110.40 &      73.51 \\
        Difficult  &   63 &  43.2 &      3.4 &      5.1 &      5.1 &      5.1 &\bf   6.8 &      1.7 &      3.4 &      5.1 &     112.14 & \bf  75.27 &     112.47 & \bf  75.30 &     111.28 &      79.14 &     115.25 &      80.67 \\
           NonOpt  &    4 &  59.9 &      --- &      --- &      --- &      --- &      --- &      --- &      --- &      --- & \bf  94.28 &     107.82 &      96.42 &     108.43 &     103.94 &     114.57 &     101.74 & \bf  94.45 \\
\hline
      Coverage 10  &  525 &  89.3 &     21.7 &     27.8 &     22.1 &     28.2 &     18.2 &     22.3 &     21.3 &\bf  34.5 &      40.88 &      28.21 &      39.65 & \bf  27.01 &      48.90 &      37.93 &      45.89 & \bf  27.75 \\
      Coverage 20  &  129 &  81.1 &     18.6 &     19.4 &     17.8 &     21.7 &     16.3 &     16.3 &     13.2 &\bf  22.5 &      52.14 & \bf  36.72 &      53.80 &      37.64 &      55.55 &      42.71 &      58.16 &      41.33 \\
      Coverage 30  &   66 &  52.0 &      7.6 &     10.6 &      7.6 &\bf  12.1 &     10.6 &\bf  12.1 &      6.1 &     10.6 &      62.92 & \bf  42.80 &      64.10 &      44.78 &      65.63 &      44.79 &      65.07 &      48.54 \\
      Coverage 40  &   52 &  56.5 &      5.8 &      9.6 &      5.8 &      7.7 &      3.8 &      3.8 &      7.7 &\bf  21.2 &       9.23 &       6.51 &       8.86 &       5.34 &      15.70 &      11.71 &       7.59 & \bf   2.01 \\
      Coverage 50  &   97 &  71.9 &     19.6 &     22.7 &     19.6 &     25.8 &     17.5 &     19.6 &     19.6 &\bf  47.4 &      20.31 &      11.62 &      20.70 &      10.24 &      21.24 &      15.07 &      18.18 & \bf   8.89 \\
      Coverage 60  &   35 &  77.8 &     40.0 &     40.0 &     31.4 &     37.1 &     31.4 &     31.4 &     40.0 &\bf  77.1 &       3.39 &       4.44 &       4.39 &       3.11 &       4.04 &       2.85 &       2.06 & \bf   0.58 \\
      Coverage 70  &   12 &  66.7 &     66.7 &\bf  91.7 &     66.7 &     83.3 &     58.3 &     83.3 &     66.7 &\bf  91.7 &       0.92 &       0.44 &       0.74 &       0.24 &       2.57 &       2.43 &       2.20 & \bf   0.01 \\
      Coverage 80  &    6 &  54.5 &    100   &    100   &    100   &    100   &    100   &    100   &    100   &    100   &       0.00 &       0.00 &       0.00 &       0.00 &       0.00 &       0.00 &       0.00 &       0.00 \\
      Coverage 90  &    4 &  28.6 &     75.0 &     75.0 &     75.0 &     75.0 &\bf 100   &     75.0 &\bf 100   &\bf 100   &       0.49 &       0.49 &       0.49 &       0.49 & \bf   0.00 &       0.49 & \bf   0.00 & \bf   0.00 \\
     Coverage 100  &    1 &   9.1 &    100   &    100   &    100   &    100   &    100   &    100   &    100   &    100   &       0.00 &       0.00 &       0.00 &       0.00 &       0.00 &       0.00 &       0.00 &       0.00 \\
\hline
              All  &  927 &  77.2 &     21.2 &     25.9 &     21.0 &     26.5 &     18.5 &     21.3 &     20.4 &\bf  34.9 &      37.67 &      25.89 &      37.35 & \bf  25.22 &      43.39 &      32.99 &      41.15 &      25.99 \\
\hline
\end{tabular}
 }}
\end{table}

For each strong approximation algorithm,
 we implemented the most promising variants,
 both of combinatorial and LP-based nature.
Thereby,
 we identified several areas to improve or extend the known algorithms either theoretically (\eg
  extending the applicability of the \texttt{gen=voronoi} strategy) or practically.
We conducted a large study of the different algorithms and their variants,
 and compared them to simple $2$-approximations and an exact algorithm.

The choice of $k=3$ turns out to be practical;
 there, the simplest and oldest below-2 approximation%
 ---Zelikovsky's $11/6$ approximation, combined with a direct $3$-restricted component generation~\citep{Z93:IPL}---%
 offers the best compromise
 between time consumption and solution quality in practice.
For higher $k$, the `relative greedy heuristic'~\citep{Z95}
 seems to be a viable choice \wrt\ solution quality.
This is surprising since the loss-contracting algorithm~\citep{RZ05}
 and the LP-based algorithm~\citep{GORZ12}
 provide better theoretical bounds.

In order to make the strong approximations more practical for higher $k$,
 it seems inevitable to find a way to significantly decrease the number
of considered full components.
 This could, \eg
 be achieved by a more efficient generation scheme
  (perhaps a generalization of the direct full component generation to $k \geq 4$),
 by removing dominated full components,
 or %
 by starting with a small number of full components and constructing further ones only when it seems fit.
All the above approaches clearly deserve further in-depth studies, both from theory and practice.

\paragraph{Acknowledgement.}
We are grateful to Matthias Woste for initial implementations.

\bibliographystyle{ACM-Reference-Format-Journals}
\bibliography{refs}

\end{document}